\documentclass{acm_proc_article-sp}
\pdfoutput=1

\usepackage{graphicx}
\usepackage{subfigure}
\usepackage{float}
\usepackage{enumerate}
\usepackage{bbm}

\def\B{{\bf B}}

\def\L{{\bf L}}

\def\S{{\bf S}}

\def\T{{\bf T}}

\def\x{{\bf x}}

\def\0{{\bf 0}}
\def\1{{\bf 1}}

\def\Ome{\mbox{\boldmath$\Omega$\unboldmath}}

\def\argmax{\mathop{\rm argmax}}
\def\argmin{\mathop{\rm argmin}}

\def\Lsize{\hbox{\space \raise-2mm\hbox{$\textstyle \L \atop \scriptstyle {m\times 3n}$} \space}}
\def\Ssize{\hbox{\space \raise-2mm\hbox{$\textstyle \S \atop \scriptstyle {m\times 3n}$} \space}}
\def\Osize{\hbox{\space \raise-2mm\hbox{$\textstyle \Ome \atop \scriptstyle {m\times 3n}$} \space}}
\def\Tsize{\hbox{\space \raise-2mm\hbox{$\textstyle \T \atop \scriptstyle {3n\times n}$} \space}}
\def\Bsize{\hbox{\space \raise-2mm\hbox{$\textstyle \B \atop \scriptstyle {m\times n}$} \space}}

\newcommand{\tabincell}[2]{\begin{tabular}{@{}#1@{}}#2\end{tabular}}

\begin{document}

\title{S-PowerGraph: Streaming Graph Partitioning for Natural Graphs by Vertex-Cut}

\numberofauthors{3} %
\author{
\alignauthor
Cong Xie\\
       \affaddr{Department of Computer Science and Engineering\\
       Shanghai Jiao Tong University \\
        800 Dong Chuan Road \\
        Shanghai, China 200240}\\
       \email{xcgoner1108@gmail.com}
\alignauthor
Wu-Jun Li\\
       \affaddr{
		National Key Laboratory for Novel Software Technology\\
        Department of Computer Science and Technology\\
		Nanjing University}\\
       \email{liwujun@nju.edu.cn} \alignauthor
Zhihua Zhang\\
       \affaddr{Department of Computer Science and Engineering \\
        Shanghai Jiao Tong University \\
        800 Dong Chuan Road \\
         Shanghai, China 200240}\\
       \email{zhang-zh@cs.sjtu.edu.cn}
}

\maketitle
\begin{abstract}
One standard solution for analyzing large natural graphs is to adopt distributed computation on clusters. In distributed computation, graph partitioning~(GP) methods assign the vertices or edges of a graph to different machines in a balanced way so that some distributed algorithms can be adapted for. Most of traditional GP methods are \emph{offline}, which means that the whole graph has been observed before partitioning. However, the offline methods often incur high computation cost. Hence, \emph{streaming graph partitioning}~(SGP) methods, which can partition graphs in an online way, have recently attracted great attention in distributed computation. There exist two typical GP strategies: edge-cut and vertex-cut. Most SGP methods adopt edge-cut, but few vertex-cut methods have been proposed for SGP. However, the vertex-cut strategy would be a better choice than the edge-cut strategy because the degree of a natural graph  in general follows a highly skewed power-law  distribution. Thus, we propose a novel method, called \emph{S-PowerGraph}, for SGP of natural graphs by vertex-cut. Our S-PowerGraph method is simple but effective. Experiments on several large natural graphs and synthetic graphs show that our S-PowerGraph can outperform the state-of-the-art baselines.
\end{abstract}

%
\category{7}{Data streams}{}
\category{13}{Graph mining}{}




\section{Introduction}

In this big data era, data mining and machine learning with massive datasets have captured more and more attentions from research community and industry. In particular, recent years have witnessed the emergence of large natural graphs in many applications like Twitter and Facebook. It is very difficult for a single commodity machine to handle such kinds of large-scale graphs. Hence, more and more works try to adopt parallel~(distributed) computation on clusters for tasks with large graphs. To complete distributed computation, we first need to partition the whole graph across the cluster. Hence, \emph{graph partitioning}~(GP)~\cite{stanton2012streaming,gonzalez2012powergraph,jain2013graphbuilder,stanton2014SODA} has become one of the hot topics in graph-related research.

In distributed computation with clusters, the inter-machine communication cost is typically high. Hence, the goal of GP is to minimize the inter-machine~(cross partition) communication, and simultaneously keep the number of vertices~(or edges) in every partition approximately balanced~\cite{stanton2012streaming}. There exist two representative GP strategies: one is edge-cut and
the other is vertex-cut. Edge-cut~\cite{stanton2012streaming,DBLP:conf/kdd/NishimuraU13,tsourakakis2012fennel} divides the graphs by cutting the edges and assigns the vertices to different partitions. To achieve the goal of GP, edge-cut tries to minimize the cross-partition edges and keep the number of vertices in every partition balanced. On the contrary, vertex-cut~\cite{gonzalez2012powergraph,jain2013graphbuilder} divides the graphs by cutting the vertices and assigns the edges to different partitions. To achieve the goal of GP, vertex-cut tries to minimize the number of replications of vertices and keep the number of edges in every partition balanced. Figure~\ref{fig:cut} illustrates the difference between these two strategies. In Figure~\ref{fig:cut}~(a), the edges $(\textsf{A,C})$ and $(\textsf{A,E})$ are cut. The two partitions are node sets $\{\textsf{A,B,D}\}$ and $\{\textsf{C, E}\}$, respectively. In Figure~\ref{fig:cut}~(b), the node $\textsf{A}$ is cut. The two partitions are edge sets $\{(\textsf{A,B}),(\textsf{A,D})\}$ and $\{(\textsf{A,C}),(\textsf{A,E})\}$, respectively. It is easy to see that there exists essential difference between these two strategies.
\begin{figure}[htbp]
\centering
\subfigure[Edge-Cut]{\includegraphics[width=0.23\textwidth]{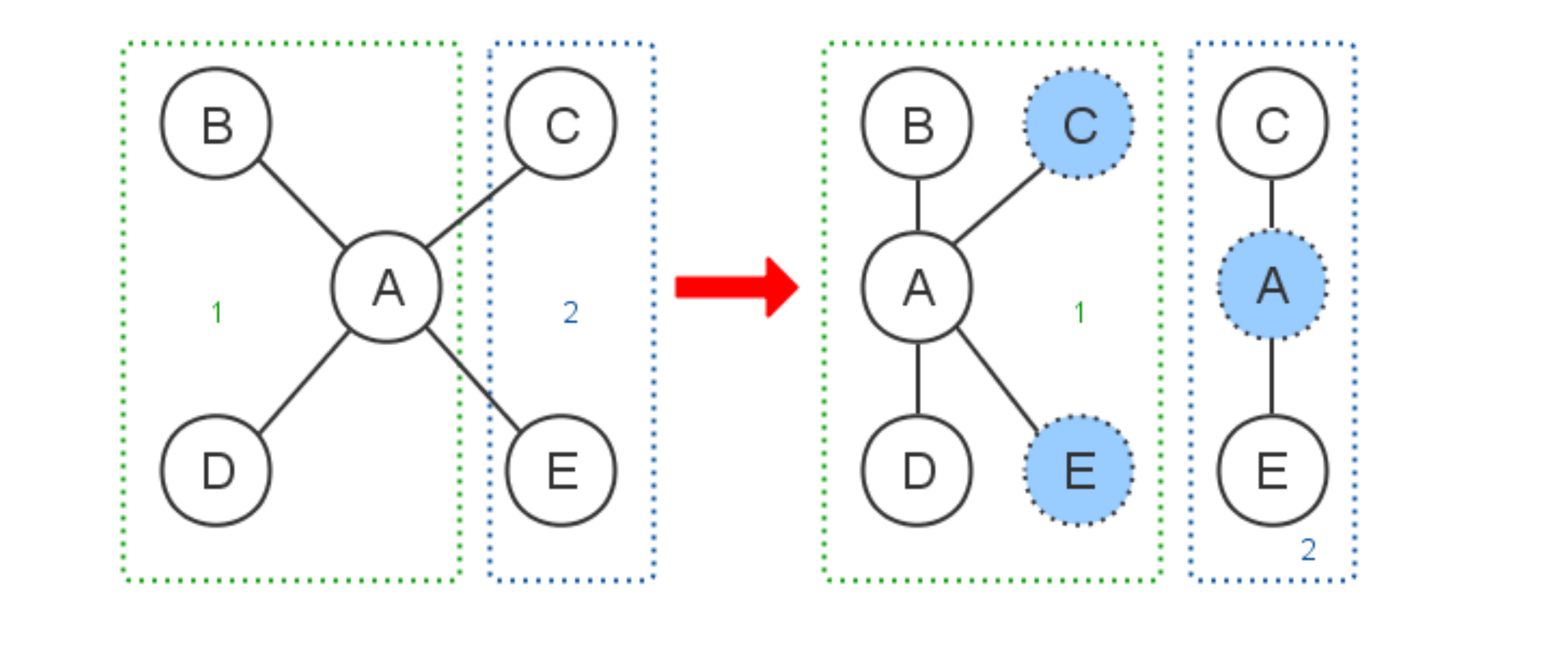}}
\subfigure[Vertex-Cut]{\includegraphics[width=0.23\textwidth]{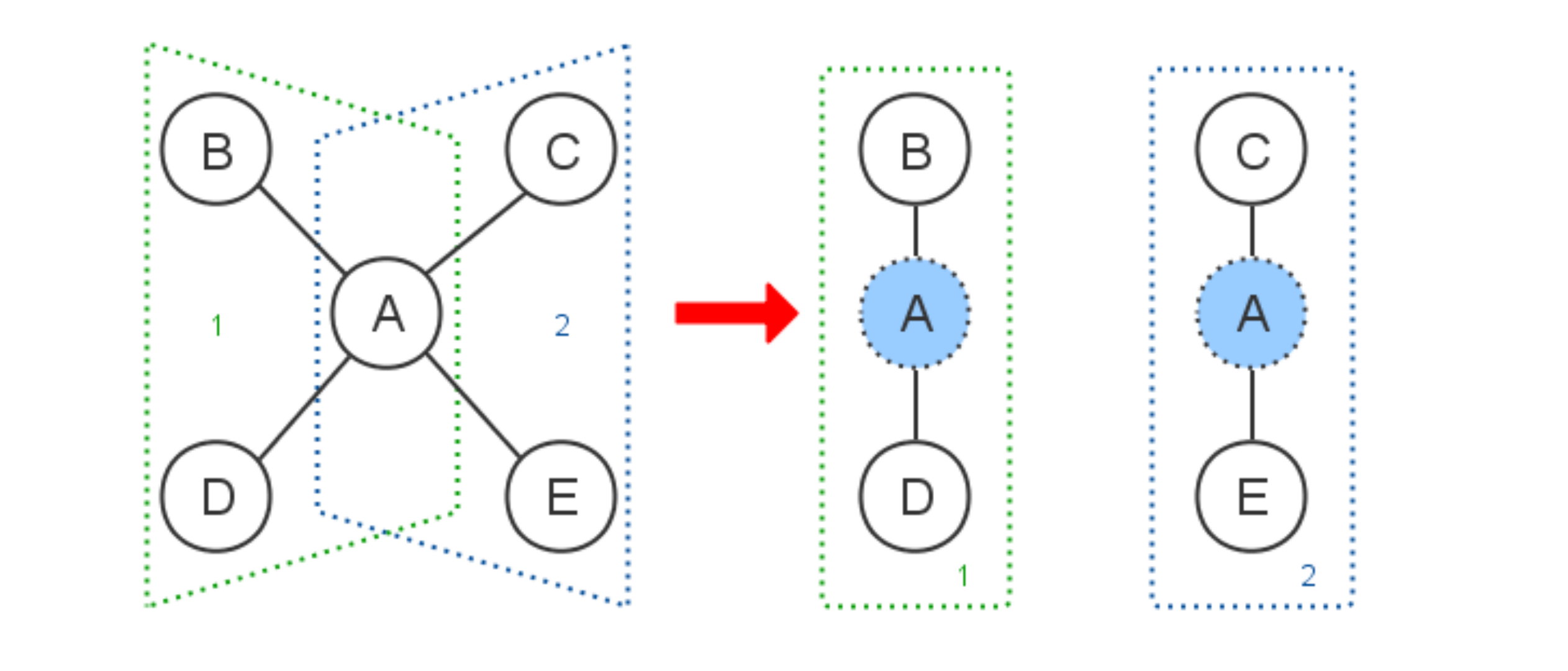}}
\vskip -0.2cm
\caption{\small Two strategies for graph partitioning. Shaded vertices are ghosts and mirrors respectively.}
\label{fig:cut}
\end{figure}

Most of the traditional GP methods are \emph{offline}~\cite{karypis1995metis}, which means that the whole graph has been observed before partitioning. However, the offline methods often incur high computation cost and hence can not scale to big data applications. Furthermore, the system often need to serially read the graph data from some disks onto a cluster or dynamically crawl the graph from some web sites. Hence, \emph{streaming graph partitioning}~(SGP) methods~\cite{stanton2012streaming,jain2013graphbuilder,DBLP:conf/kdd/NishimuraU13,tsourakakis2012fennel}, which can partition graphs in an online way, have recently attracted more and more attention from researchers. Representative methods include linear deterministic greedy~(LDG)~\cite{stanton2012streaming} and FENNEL~\cite{tsourakakis2012fennel}. However, most existing SGP methods adopt the edge-cut strategy for partitioning.

Actually, many existing offline~(non-streaming) works, such as PowerGraph~\cite{gonzalez2012powergraph}, have found that the vertex-cut strategy can outperform edge-cut in applications with natural graphs which in general have highly skewed power-law degree distributions. Even in the edge-cut based SGP work~\cite{stanton2012streaming}, the authors admitted that the \emph{edge-balanced partition}, which can be achieved by vertex-cut, may be more preferable than edge-cut for power-law distributed graphs. However, few vertex-cut methods have been proposed to solve the SGP problem.


In this paper, a novel method called \emph{S-PowerGraph}\footnote{The \emph{S} in \emph{S-PowerGraph} comes from the term \emph{streaming}.} is proposed for SGP of natural graphs by vertex-cut.  The main contributions of this paper are briefly outlined as follows:
\begin{itemize}
\item This is the first work to systematically study and compare different heuristics for SGP of natural graphs by vertex-cut.
\item Although GraphBuilder~\cite{jain2013graphbuilder} and PowerGraph~\cite{gonzalez2012powergraph} are not motivated by SGP and hence cannot be directly used for SGP, we propose methods to adapt them for SGP.
\item  S-PowerGraph can be viewed as an improved version of the \emph{adapted PowerGraph}, by making use of the property of power-law degree distribution for vertex-cut based \mbox{SGP}.
\item Experiments on several large natural graphs and synthetic graphs show that our S-PowerGraph can outperform other baselines, including the adapted \mbox{GraphBuilder} and \mbox{PowerGraph}, to achieve the state-of-the-art performance in real applications.
\end{itemize}


\section{Problem Formulation}


Graph $G$ is denoted by \(G = (V, E)\) where \(V = \{v_1, v_2, \ldots, v_n\}\) is a set of vertices and  \(E \subseteq V \times V\) denotes a set of edges. This implies that the number of vertices $|V|=n$, i.e., the cardinality of $V$ is $n$.
We say that $v_i$ and $v_j$ are neighbors if \((v_i, v_j) \in E\).
Typically, many graph mining and learning algorithms can be expressed by computations on vertices of a graph, and the edges of such a graph represents dependency between the data samples stored on the vertices~\cite{low2010graphlab,gonzalez2012powergraph}.

The goal of GP is to \emph{balance} the parallel work on each machine and \emph{minimize} the communication cost among different machines.
In the following, we will briefly introduce two important GP strategies: edge-cut and vertex-cut. Furthermore, we will also discuss the SGP problem based on vertex-cut.

\subsection{GP by Edge-Cut}

The traditional GP problem is to balance the number of vertices in each partition and minimize the edges across different partitions, which is known as edge-cut based GP~\cite{karypis1995metis,stanton2012streaming,DBLP:conf/kdd/NishimuraU13,stanton2014SODA,tsourakakis2012fennel}. More specifically, the problem is defined as follows:
\begin{align*}
& \min \limits_{A} \; \big\vert\{e \mid e = (v_i, v_j) \in E, v_i \in V_x, v_j \in V_y, x \neq y\} \big\vert \\
& s.t. \; \frac{\max \limits_{i} \vert V_i \vert}{\frac{1}{p} \sum_{i=1}^p \vert V_i \vert} \leq 1 + \xi,
\end{align*}
where $|\cdot |$ denotes the cardinality of a set, $\xi$ is a constant with $0 \leq \xi \leq 1 $, $c$ is the number of partitions,
$A =\{V_1, V_2, \ldots, V_{p}\}$ is a partition of $V$, i.e., $V_1 \cup \cdots \cup V_{p} =V$ and $V_x \cap V_y=\emptyset $ for $x\neq y$.

The objective function in the above formulation tries to minimize the number of edges across different partitions, and the constraint makes the partitions as balanced as possible. The demonstration of edge-cut is shown in Figure~\ref{fig:cut}~(a). Representative methods include the offline method Metis~\cite{karypis1995metis} and the streaming methods LDG~\cite{stanton2012streaming} and FENNEL~\cite{tsourakakis2012fennel}.

\subsection{GP by Vertex-Cut }

The edge-cut methods are efficient when partitioning sparse graphs, but they are not suitable when adapted for scale-free graphs.
Typically, natural graphs are not only large but also skewed.
The degree of the graph usually follows the power-law degree distribution:
\[
P(d) \propto d^{-\alpha},
\]
where $P(d)$ is the probability that a vertex has degree $d$, \(\alpha\) is a positive constant. For example, $ \alpha \approx 2.2 $ in the twitter-2010 social network~\cite{kwak2010twitter}. This skewed degree distribution might result in an imbalanced partitioning and challenge the traditional edge-cut methods.

PowerGraph~\cite{gonzalez2012powergraph} defined a balanced $p$-way vertex-cut based graph partitioning model. Let each edge $e \in E$ be assigned to a machine $M(e) \in \{1, \ldots, p\}$.
Then each vertex spans a set of different partitions $ A(v) \subseteq \{1, \ldots, p\} $. Thus, $\vert A(v) \vert$ is the number of replications of $v$ among different machines. The balanced vertex-cut based partitioning model is defined as follows:
\begin{align*}
& \min \limits_A \frac{1}{\vert V \vert} \sum \limits_{v \in V} \vert A(v) \vert
\\
& s.t. \; \max \limits_m \vert \{ e \in E \mid M(e) = m \} \vert < \lambda\frac{\vert E \vert}{p},
\end{align*}
where $\lambda \geq 1$ is an imbalance factor.

In vertex-cut, the number of edges on each machine is used to estimate the computation cost of that machine, and the number of the replications of the vertices is used to estimate the communication cost. The objective is still to balance the workload among machines and minimize the communication cost.




\subsection{SGP by Vertex-Cut}


The problem of SGP by vertex-cut tries to optimize the following conditional expectation:
\begin{align} \label{eqn:stream1}
&\argmin \limits_k \; \mathbb{E} \Big[ \sum \limits_{v \in V} \vert A(v) \vert \Bigg\vert M_i, M(e_{i+1}) = k \Big], \\
& s.t. \; \max \limits_m \vert \{ e \in E \mid M(e) = m \} \vert < \lambda\frac{\vert E \vert}{p}, \nonumber
\end{align}
where $\mathbb{E}$ denotes the expectation operation, $k$ is the index number of a partition, \(M_i\) is the assignment of the previous \(i\) edges, $e_{i+1}$ is the current edge in the stream.


\section{Methodology}
\label{sec:proposal}



We consider a streaming model in which edges arrive in a stream and each edge appears only once in the stream.
The vertices also arrive in a stream order just like what a web crawler does. Moreover, each time a vertex arrives together with some of its edges.


For SGP, the order that the vertex arrives is very important because it can influence the performance of the algorithms significantly~\cite{stanton2012streaming,tsourakakis2012fennel}. In our streaming setting, the vertices arrive in a particular order with all of their out-edges arriving in a random order. As stated in existing works~\cite{stanton2012streaming,tsourakakis2012fennel}, there are three typical orders in which the vertices arrive:
\begin{itemize}
\item \textbf{Random~(Rnd)}: The vertices simply arrive in a random order or in an order given by a random permutation of the vertices.
\item \textbf{Breadth-first search~(BFS)}: A starting vertex is selected from the graph and
new vertices arrive if they are found by a breadth-first search.
\item \textbf{Depth-first search~(DFS)}: This order is just the same as the BFS ordering except that depth-first search will be used.
\end{itemize}

\subsection{Model}

We observe that the replication factor, which is defined as the total number of replications divided by the total number of vertices, will be smaller if we cut more vertices with high degree. That is because when the degrees of vertices are under the power-law distribution, there exist many vertices with small degree and a relatively small amount of vertices are dense. So if  choices are available, we always choose to cut a vertex with relatively higher degree. Then we can simply use this strategy to improve the naive random (hashed) vertex-cut.

Relative to the model in (\ref{eqn:stream1}), we add a   constraint to guarantee if either vertex of the edge can be cut, the one with higher degree will be chosen. 
In particular, we defined the following model for streaming partitioning:
\begin{equation} \label{eqn:stream2}
\argmin \limits_{k} \; \mathbb{E} \Big[ \sum \limits_{v \in V} \vert A(v) \vert \Big\vert M_i, M(e_{i+1}) = k, k \in S \Big].
\end{equation}
Here $S$ is a set of partitions which the vertex with higher degree of $e_{i+1}$ belongs to. If neither of the two vertices of $e_{i+1}$ belongs to any partition, we let $S$  simply contain all the $k$ partitions. Thus, we define $S$ as:
\[
S =
\begin{cases}
   \{1,...,p\} &\mbox{if $A(u) = 0$ or $A(v) = 0$}\\
   A(\hat{w}) &\mbox{otherwise},
\end{cases}
\]
where
\[ \hat{w} =\argmin\limits_{w}  D(w) \; \mbox{ s.t. } \;  w \in \{u, v\} \; \mbox{ for } \;  (u,v) = e_{i+1}.
\]
Here $D(u)$ is the degree of vertex $u$. In this paper, we simply consider the graph undirected. For particular applications, $D(u)$ can be the degree for undirected graph or directed graph, and it can further be either in-degree or out-degree of $u$ if the graph is directed.

Additionally, since we do the partitioning job in a real on-line condition, the degree (in-degree and out-degree) of each vertex is not known in advance. The degree distribution is unveiled gradually when the edges of the graph arrive in stream. So we need to find a way to estimate the degree distribution of the whole graph.

We use model in~(\ref{eqn:stream2}) and some kind of estimation of degree distribution to build the algorithm \textit{Random-Degree}, \textit{Degree} and \textit{DegreeIO} in the next section.

\subsection{Algorithms}

%
%
%
%
%

In this paper, we examine multiple partitioning algorithms. We define each of them as follows.

Let $P(k)$ denote the current set of edges belong to the $k$th partition, and let $[p]=\{1, \ldots, p\}$. 
We define
\[
maxedges = \max \limits_{k \in [p]} \{ \vert P(k) \vert \} \; \mbox{ and } \;  minedges = \min \limits_{k \in [p]} \{ \vert P(k) \vert \}.
\]
We refer to a certain partition by its index $idx$. Then the $idx$th partition is denoted by $P_{idx}$. For each algorithm, and index $idx$ will be decided by different heuristic and each edge will be assigned to the $idx$th partition. For the greedy algorithms, when assigning an edge, each algorithm will first compute a $Score$ for each partition. It is the partition with the highest $Score$ that the edge will finally be assigned to.

The first three algorithms are adapted from some related works.

\textbf{Random}
The random method simply assigns each edge via a random hash function as implemented in PowerGraph. So it is the most naive algorithm.
The edge $e$ is assigned to $P_{idx}$ where $idx$ is decided by:
\[
idx = hash(e)
\]
$hash(e)$ here is a randomized hash function.

\textbf{Grid}
GraphBuilder~\cite{jain2013graphbuilder} provides a  Grid-based Constrained Random Vertex-cut algorithm, which can also be adapted for streaming graph partitioning. In this Grid partitioning, a vertex is mapped into a certain shard which contains a constrained set of partitions. Then, the assignment will be chosen from the constrained set randomly. The constraint on the choice of partitions produces less replications than the naive randomized partitioning.\\
The edge $e$ is assigned to $P_{idx}$ where $idx$ is decided by:
\[
idx = Grid\_hash(e)
\]
The hash function $Grid\_hash(e)$ here is the same one implemented by GraphBuilder.

\textbf{Balance~(i.e., PowerGraph)} - An enhanced version of the greedy algorithm proposed in PowerGraph \cite{gonzalez2012powergraph}. However some particular orders may result in relatively imbalance partitioning for the algorithm introduced by PowerGraph. So in practice the proposed greedy algorithm cannot be used as a baseline. We instead need an algorithm to enhances the constraint of imbalance which turns out to be the algorithm \textit{Balance(PowerGraph)}. A new constraint is added to avoid imbalance. This modified version  plays a role as baseline which demonstrates the performance of the PowerGraph partitioning.
\begin{align*}
 Score(k) &=  \mathbf{1} \{ k \in A(u) \} + \mathbf{1} \{ k \in A(v) \} + balance(k), \\
balance(k) &= \dfrac{maxedges - \vert P(k) \vert}{maxedges - minedges + 1}, \\
(u,v) &= e, \\
k &\in  [p].
\end{align*}
where $\mathbf{1}\{\#\}=1$ if $\#$ is true and  $\mathbf{1}\{\#\}=0$ otherwise.

If $ \frac{\max \limits_{i} \vert p_i \vert}{avg \vert p_i \vert} \geq 1.1 $, then the edge $e$ is assigned to $P_{idx}$ where $idx$ is decided by:
\[
idx = \argmax \limits_{k \in [p]} \{ balance(k) \}
\]
Otherwise, the edge $e$ is assigned to $P_{idx}$ where $idx$ is decided by:
\[
idx = \argmax \limits_{k \in [p]} \{ Score(k) \}
\]

The following algorithms are proposed by us which are based on the model in (2).

\textbf{Random-Degree}
The edge $e$ is assigned to $P_{idx}$ where $idx$ is decided by:
\[
idx = hash(v_L)
\]
where
\[ v_L =\argmin\limits_{w}  D_{in}(w) \; \mbox{ s.t. } \;  w \in \{u, v\}
\]

\textbf{Degree} - Assign $e$ to a partition of maximal score, breaking ties randomly. The score of the $k$th partition is defined as follows:
\begin{align*}
Score_{in}(k) &=  \mathbf{1} \{ k \in A(u) \} + \mathbf{1} \{ k \in A(v) \} \\
 & \quad +  \mathbf{1} \{ k \in A(u), D_{in}(u) \leq D_{in}(v) \} \\
 & \quad +  \mathbf{1} \{ k \in A(v), D_{in}(v) \leq D_{in}(u) \} \\
 & \quad +  balance(k), \\
balance(k) &= \dfrac{maxedges - \vert P(k) \vert}{maxedges - minedges + 1}, \\
(u,v) &= e, \\
k &\in [p].
\end{align*}
If $ \frac{\max \limits_{i} \vert p_i \vert}{avg \vert p_i \vert} \geq 1.1 $, then the edge $e$ is assigned to $P_{idx}$ where $idx$ is decided by:
\[
idx = \argmax \limits_{k \in [p]} \{ balance(k) \}
\]
Otherwise, the edge $e$ is assigned to $P_{idx}$ where $idx$ is decided by:
\[
idx = \argmax \limits_{k \in [p]} \{ Score_{in} \}
\]

\textbf{DegreeIO} - The score of the $k$th partition is defined as follows:
\begin{align*}
Score_{in}(k) &=  \mathbf{1} \{ k \in A(u) \} + \mathbf{1} \{ k \in A(v) \} \\
& \quad +  \mathbf{1} \{ k \in A(u), D_{in}(u) \leq D_{in}(v) \} \\
& \quad +  \mathbf{1} \{ k \in A(v), D_{in}(v) \leq D_{in}(u) \} \\
& \quad +  balance(k) \\
Score(k) &=  \mathbf{1} \{ k \in A(u) \} + \mathbf{1} \{ k \in A(v) \} \\
 & \quad +  \mathbf{1} \{ k \in A(u), D(u) \leq D(v) \} \\
& \quad +  \mathbf{1} \{ k \in A(v), D(v) \leq D(u) \} \\
& \quad +  balance(k), \\
balance(k) &= \dfrac{maxedges - \vert P(k) \vert}{maxedges - minedges + 1}, \\
(u,v) &= e, \\
k  &\in [p].
\end{align*}
For the current edge $e$, we let $u$ be the source vertex and $v$ be the target vertex.
If $D_{out}(u) \geq p$ and $D_{out}(v)$ is unknown which means vertex $v$ has not arrived in the stream,
then we do not assign the edge immediately and store it in a buffer until vertex $v$ arrives.
If $D_{out}(v)$ is known, the edge $e$ is assigned to $P_{idx}$ where $idx$ is decided by:
\[
idx = \argmax \limits_{k \in [p]} \{ Score(k) \}
\]
Otherwise, the edge $e$ is assigned to $P_{idx}$ where $idx$ is decided by:
\[
idx = \argmax \limits_{k \in [p]} \{ Score_{in}(k) \}
\]


\subsection{Theoretical Analysis}
Different algorithms have different motivations.

The first three algorithms are from related works and relatively need less computation and memory.

For all the greedy algorithms \textit{Balance(Powergraph)}, \textit{Degree} and \textit{DegreeIO}, when the graph is not very skewed, \textit{Balance} is a good choice which requires less computation and memory.

\textit{Random-Degree} demonstrates the main idea of our model: decisions of which partition to be assigned to is always according to the vertex with lower degree. Although this algorithm cannot be used in the stream setting because the degree of each vertex is unknown in a stream,
a theoretical analysis is available to show how our idea reduce the replication factor.

\newtheorem{theorem}{Theorem}
\begin{theorem}
A random-degree vertex-cut on $p$ machines has the expected replication factor:
\begin{align*}
&\mathbb{E}\left[ \frac{1}{\vert V \vert}\sum\limits_{v \in V} \vert A(v) \vert \right] \\
& = 1 + \frac{(p-1)}{\vert V \vert} \sum\limits_{v \in V} \left(1-(1-\frac{1}{p})^{(1-Ratio(v))D[v]}\right)
\end{align*}
Here $D(v)$ is the degree of vertex v. And all edges are treated as undirected. $Ratio(v)$ is a ratio of degree of each vertex.
\end{theorem}

\begin{proof}
Our degree-random vertex-cut algorithm:
for edge $e(u, v)$, $v_L = arg\min\limits_w \{D(w) \vert w \in \{u,v\} \}$ which is the vertex with lower degree.
Then we assign the edge by some kind of hash function.
So we have the assignment of $e$:
\[
A(e) = hash(v_L),
\]

Denote the indicator $P_i$ as the event that vertex $v$ has at least one of $S(v)$ in the $i$th partition where $S(v)$ is a subset of the adjacent edges of vertex $v$. Then follow the above procedure, the expectation $\mathbb{E}[P_i]$ is:
\[
\mathbb{E}[P_i] = 1-(1-\frac{1}{p})^{D_s[v]}
\]
where $D_s(v)$ is the number of edges in $S(v)$.

We denote the fraction of the adjacent edges of $v$ that satisfy $D(v) \leq D(u)$ as $Ratio(v)$ where $u$ is the another vertex of each adjacent edge $e(u,v)$.

Intuitively, if we assume $P(d) = Cd^{-\alpha}$ I have:
\[
1-Ratio(v) = P\left( d \leq D(v) \right) = \sum\limits_{d=1}^{D(v)}P(d) = C\sum\limits_{d=1}^{D(v)}d^{-\alpha}
\]

For some vertex $v$, $Ratio(v)$ of its adjacent edges will be hashed to one certain partition, and the residual will be randomly hashed to all the $p$ partitions.
So for the $Ratio(v)$ of the adjacent edges, the number of replication of $v$ is simply 1.

And for the residual $1-Ratio(v)$ of the adjacent edges, we have $\mathbb{E}[P_i] = 1-(1-\frac{1}{p})^{(1-Ratio(v))D[v]}$. Here $P_i$ will the other $p-1$ partitions than the one that already have a replication.

Put the two parts together, then we can have:
\begin{align*}
&\mathbb{E} \left[ \vert A(v) \vert \right] \\
& = 1 + \sum\limits_{i=1}^{p-1}\mathbb{E}[P_i] \\
& = 1 + (p-1)\left(1-(1-\frac{1}{p})^{(1-Ratio(v))D[v]}\right)
\end{align*}

Thus the expected replication factor is:
\begin{align*}
&\mathbb{E}\left[ \frac{1}{\vert V \vert}\sum\limits_{v \in V} \vert A(v) \vert \right] \\
& = \frac{1}{\vert V \vert} \sum\limits_{v \in V} \left( 1 + (p-1)\left(1-(1-\frac{1}{p})^{(1-Ratio(v))D[v]}\right) \right) \\
& = 1 + \frac{(p-1)}{\vert V \vert} \sum\limits_{v \in V} \left(1-(1-\frac{1}{p})^{(1-Ratio(v))D[v]}\right)
\end{align*}
\end{proof}

The two novel greedy algorithms \textit{Degree} and \textit{DegreeIO} can be viewed as a greedy version of \textit{Random-Degree} in the stream setting.

\textit{Degree} is a greedy procedure which makes use of the in-degree distribution. The idea seems similar to that proposed in PowerLyra~\cite{chen2013institute} which acts as a component of the distributed framework. However, the motivations of both are different.
Unlike PowerLyra, we do not assume that the application adapted for the graph is known; i.e.,  we do not know which kind of edges (in-edges or out-edges) the computation involves in advance. Thus, we  treat each edge as undirected. But the current problem is that the out-degree of a vertex is unknown until it arrives with all its out-edges. Fortunately,  the degree distribution of the in-edges or the relationship of the degrees between vertices can be easily estimated by the vertices arrived according to our streaming model. This is why we only consider the in-degree in this step.

\textit{DegreeIO} aims to meet the challenge that both in-degree and out-degree of natural graphs are skewed and to further improve the vertex-cut. To solve the problem that the out-degree of some vertex is unknown temporarily, the edge is buffered until the out-degree is known. For the vertex whose out-degree is small, we simply reuse the \textit{Degree} step to reduce the memory use.

So if we are sure in advance that the in-degree of the graph is more skewed than the out-degree, \textit{Degree} is a better choice with less computation and memory use. Otherwise \textit{DegreeIO} is always a better choice with slightly more computation cost. The \emph{Degree} and \emph{DegreeIO} are two variants of our \mbox{S-PowerGraph} method.


\section{Experiment}
In this section we demonstrate the setup and results of our experiments as well as the discussion of the results. 

We use a server with Intel Xeon cores, 64GB of RAM.

\subsection{Datasets}
The collection of large graph datasets we use in the experiments consists of real-world graphs shown in Table~\ref{tab:realtable} and synthetic power-law graphs shown in Table~\ref{tab:synctable}. Some of the real-world graphs were selected to be the same as the experiment of PowerLyra\cite{chen2013institute}. And some additional real-world graphs were from the UF Sparse Matrices Collection\cite{davis2011university}. The synthetic power-law graphs were generated by a combination of two synthetic graphs. For one of them, the out-degree of each vertex is sampled randomly from a Zipf distribution\cite{adamic2002zipf} with parameter $\alpha$ and the in-degree of each vertex is under a uniform distribution. For the other one, the in-degree of each vertex is sampled randomly from a Zipf distribution with parameter $\beta$ and the out-degree of each vertex is under a uniform distribution. The power-law constant ($\alpha$ and $\beta$) of the synthetic graphs ranges from 1.9 to 2.2. The smaller the $\alpha$ is, the skewer the graph will be. However the different power-law constant result in different partitioning performance. And the relationship between the distribution of in-degree and out-degree also acts as an important role. To demonstrate such fact, we separate the synthetic graph collection into two parts: the one with parameter $\alpha$ no smaller than $\beta$ shown in Table~\ref{tab:synctable}~(a) and the one with parameter $\beta$ greater than $\alpha$ shown in Table~\ref{tab:synctable}~(b).
\begin{table}[htbp]
\caption{Real-world graphs}
\centering
\subtable{

\begin{tabular}{|l|l|r|r|} \hline
Alias	& Graph		& $\vert V \vert$	& $\vert E \vert$ \\ \hline
Tw	& Twitter\cite{kwak2010twitter}		& 42M				& 1.47B \\ \hline
It	& It-2004\cite{davis2011university}		& 41M				& 1.15B \\ \hline
UK	& UK-2005\cite{boldi2004webgraph}		& 40M				& 936M \\ \hline
Arab	& Arabic-2005\cite{davis2011university}		& 22M				& 0.6B \\ \hline
Wiki	& Wiki\cite{boldi2004webgraph}		& 5.7M				& 130M \\ \hline
LJ	& LiveJournal\cite{mislove2007measurement}	& 5.4M				& 79M \\ \hline
HW	& Hollywood-2009\cite{davis2011university}	& 1.1M				& 57M \\ \hline
In	& In-2004\cite{davis2011university}	& 1.4M				& 17M \\ \hline
As	& As-skitter\cite{davis2011university}	& 1.7M				& 11M \\ \hline
CP	& CoPapersCiteseer\cite{davis2011university}	& 0.4M				& 16M \\ \hline
WG	& WebGoogle\cite{leskovec2011stanford}	& 0.9M				& 5.1M \\ \hline
\end{tabular}
\label{tab:realtable}

}
\end{table}

\begin{table}[htbp]
\caption{Synthetic graphs with different power-law constants}
\centering
\subtable[Synthetic graphs with $\alpha \geq \beta$]{
\begin{tabular}[t]{|l|r|r|r|} \hline
Alias	& $\alpha$ & $\beta$ & \#Edges \\ \hline
S1	& 2.2 & 2.2 & 71,334,974 \\ \hline
S2	& 2.2 & 2.1 & 88,305,754 \\ \hline
S3	& 2.2 & 2.0 & 134,881,233 \\ \hline
S4	& 2.2 & 1.9 & 273,569,812 \\ \hline
S5	& 2.1 & 2.1 & 103,838,645 \\ \hline
S6	& 2.1 & 2.0 & 164,602,848 \\ \hline
S7	& 2.1 & 1.9 & 280,516,909 \\ \hline
S8	& 2.0 & 2.0 & 208,555,632 \\ \hline
S9	& 2.0 & 1.9 & 310,763,862 \\ \hline
\end{tabular}

}
\\
\subtable[Synthetic graphs with $\alpha < \beta$]{
\begin{tabular}[t]{|l|r|r|r|} \hline
Alias	& $\alpha$ & $\beta$ & \#Edges \\ \hline
S10	& 2.1 & 2.2 & 88,617,300 \\ \hline
S11	& 2.0 & 2.2 & 135,998,503 \\ \hline
S12	& 2.0 & 2.1 & 145,307,486 \\ \hline
S13	& 1.9 & 2.2 & 280,090,594 \\ \hline
S14	& 1.9 & 2.1 & 289,002,621 \\ \hline
S15	& 1.9 & 2.0 & 327,718,498 \\ \hline
\end{tabular}

}
\label{tab:synctable}
\end{table}

\subsection{Baseline and Evaluation Metric}

\subsubsection{Baseline}
Table~\ref{baselinetable} lists several representative graph partitioning methods and our S-PowerGraph method. Note that the \emph{Degree} and \emph{DegreeIO} are two variants of our S-PowerGraph. In our experiment, we take algorithm \textit{Random}, \textit{Grid}, and \textit{Balance(PowerGraph)} as the baseline. Note that these algorithms comes from PowerGraph and GraphBuilder listed in Table~\ref{baselinetable}. The reason why the rest of them are not considered is that the first two are of edge-cut which cannot be compared to the results of vertex-cut and the algorithm of PowerLyra cannot be used in the stream setting.
\begin{table*}[ht]
\caption{\small Some representative graph partitioning methods.}
\begin{center}
\begin{tabular}{|r|c|c|c|c|c|c|}
\hline
Method & LDG& FENNEL         & PowerGraph   & PowerLyra       &GraphBuilder                 & S-PowerGraph  \\ \hline \hline
Partition Strategy        & edge-cut  &edge-cut      & vertex-cut  & \tabincell{c}{hybrid edge-cut\\/vertex-cut}               & vertex-cut      & vertex-cut  \\ \hline
Balance Type           &vertex-balance  &vertex-balance & edge-balance & edge-balance                    & edge-balance  & edge-balance \\ \hline
Streaming/Online       & yes &yes            & no         & no                              & no  & yes          \\ \hline
\end{tabular}
\end{center}
\label{baselinetable}
\end{table*}

\subsubsection{Evaluation Metric}
The result of the experiments can be described by three factors: communication cost, imbalance and sensibility to number of partitions. We used the replication factor to define the communication cost between different partitions which was introduced by PowerGraph\cite{gonzalez2012powergraph}:
\[
\lambda = \frac{1}{\vert V \vert} \sum \limits_{v \in V} \vert A(v) \vert
\]
, where $V$ is the set of all the vertices, $\vert V \vert$ denotes the number of vertices and $\vert A(v) \vert$ denotes the number of mirrors (replications) of vertex $v$. We experiment the replication factor on all the real-world graphs and synthetic graphs.

To make the result more clear, wee also illustrate the improvement of our algorithms over the baseline. In all the three algorithms which act as the baseline, the algorithm \textit{Balance(PowerGraph)} enjoys the lowest replication factor defined above. So we compare the replication factor of our algorithms to that of \textit{Balance(PowerGraph)} and denote the improvement as:
\[
improvement = \frac{\lambda_{PowerGraph} - \lambda_{Algorithm}}{\lambda_{PowerGraph}} \times 100\%
\]
where $\lambda_{Algorithm}$ is denoted by the replication factor produced by a certain algorithm.

The second factor is the imbalance factor defined as:
\[
\rho =
\frac{\max \limits_m \vert \{ e \in E \mid A(e) = m \} \vert}{\vert E \vert / p}
\]
, where $\vert A(e) \vert$ denotes the partition that edge $e$ is placed in , $\vert E \vert$ is the number of edges of the graph and $p$ is the number of partitions. However, by using the vertex-cut strategy, the imbalance is no longer a big deal. So we just show the results of the real-world graphs.

Finally we tested the sensibility to the number of partitions of our algorithm. We used the twitter-2010 social network\cite{kwak2010twitter} and compared the replication factor by the increasing number of partitions. And we also show that the algorithm is scalable with the power-law constant by taking experiments on synthetic graphs with different power-law constant.

\subsection{Experimental Results}
In Figure~\ref{fig:real}~{(a),(b) and (c)}, we compared the replication factor on different real-world graphs in different orders by different algorithms. In all the experiments the number of partitions is 48. Our algorithm \textit{DegreeIO} produced lower replication factor in each experiment.
Figure~\ref{fig:real}~{(d),(e) and (f)} shows the improvement of the algorithms \textit{Degree} and \textit{DegreeIO} compared to \textit{Balance(PowerGraph)}.

\begin{figure*}[htb]
\centering
\subfigure[Rnd order]{\includegraphics[width=0.32\textwidth]{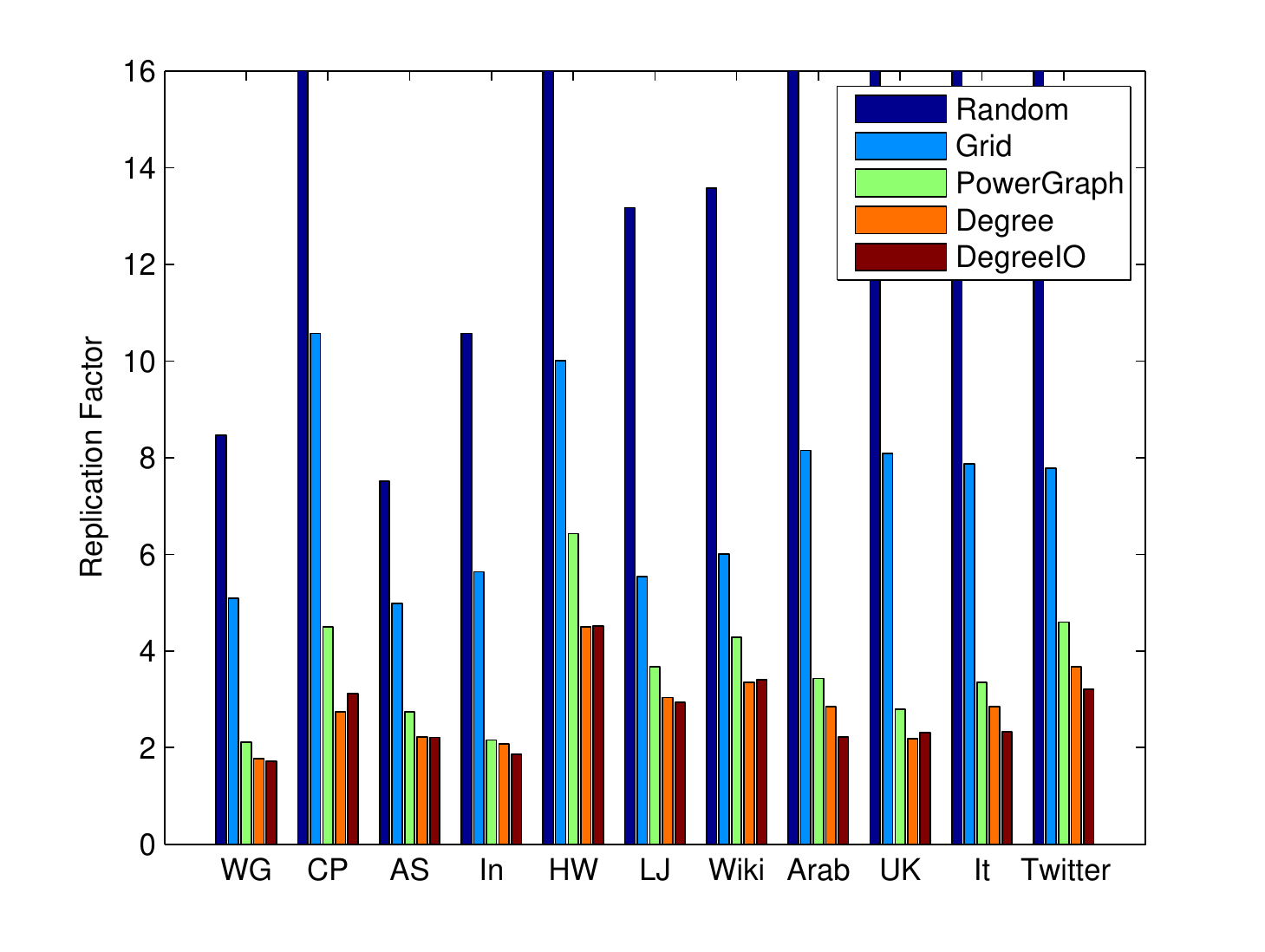}}
\subfigure[BFS order]{\includegraphics[width=0.32\textwidth]{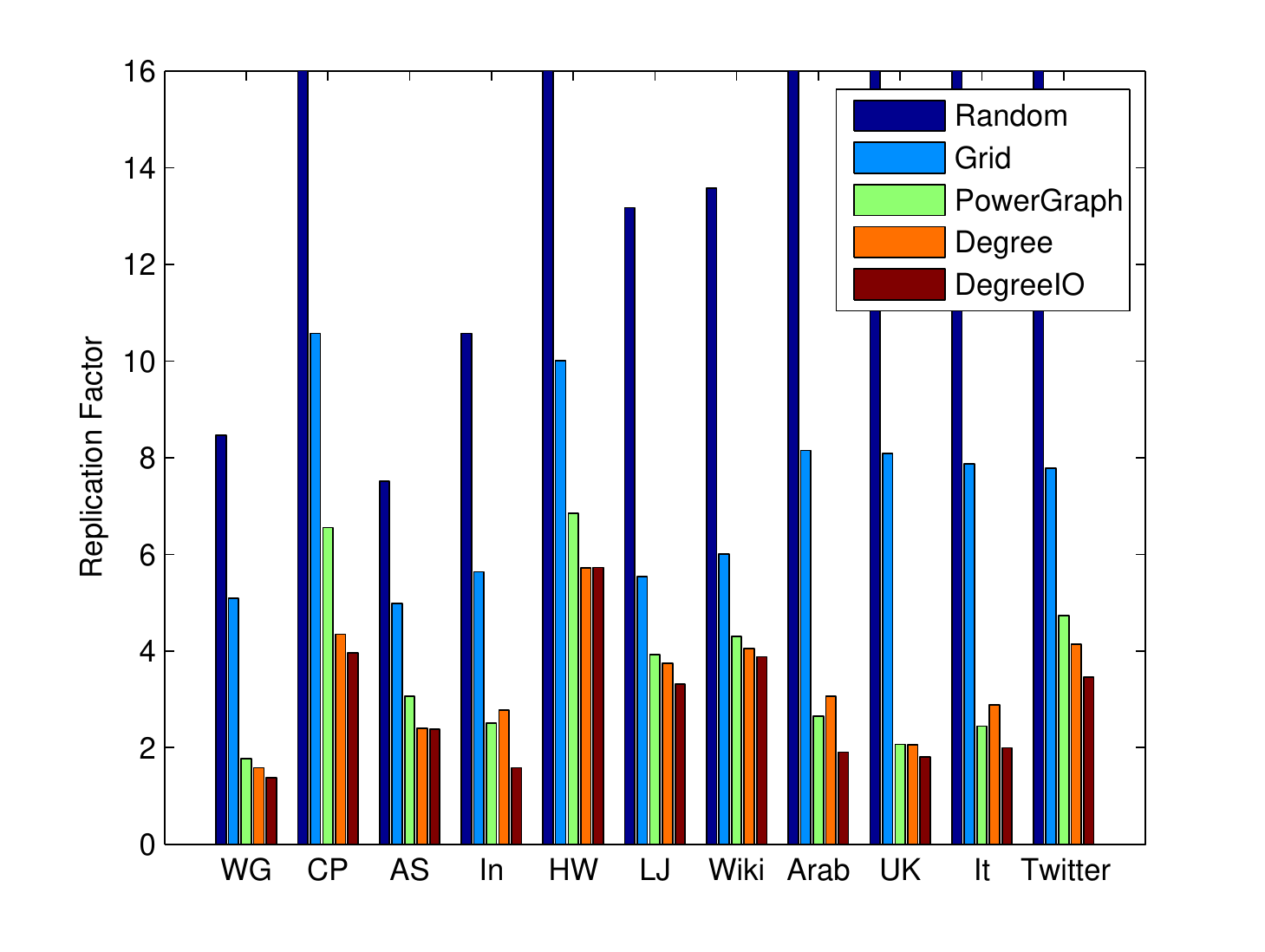}}
\subfigure[DFS order]{\includegraphics[width=0.32\textwidth]{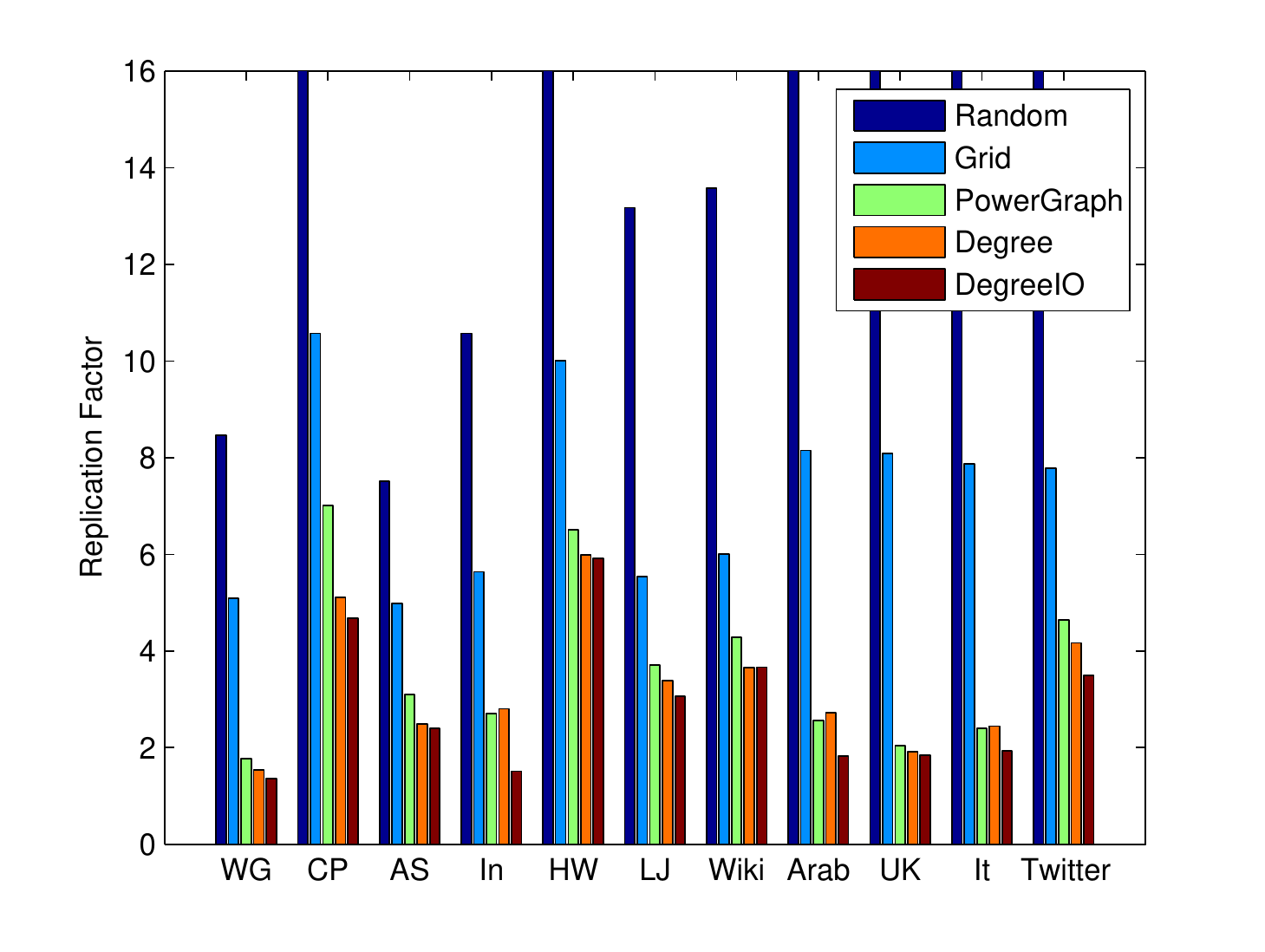}} \\
\subfigure[Rnd order]{\includegraphics[width=0.32\textwidth]{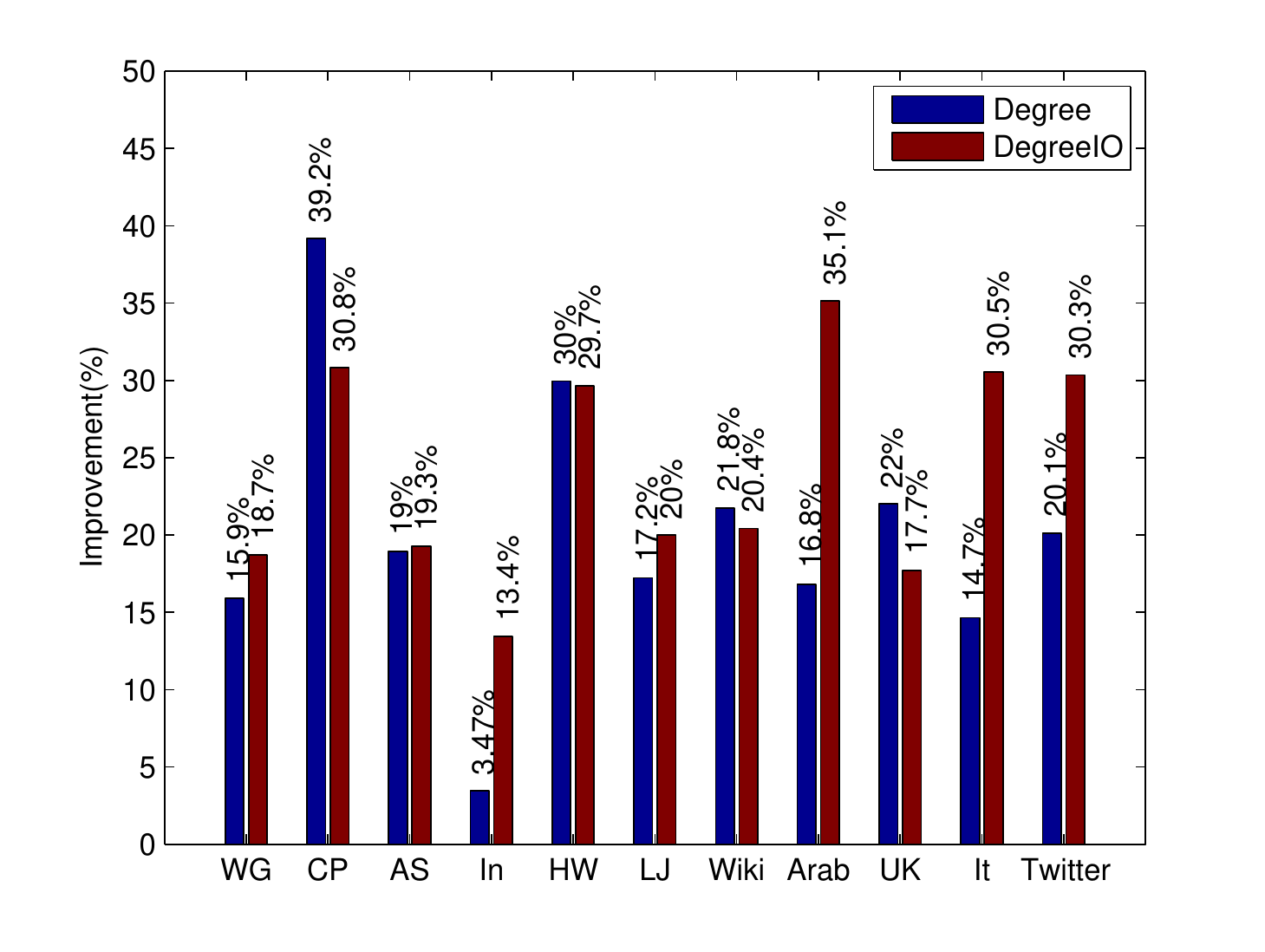}}
\subfigure[BFS order]{\includegraphics[width=0.32\textwidth]{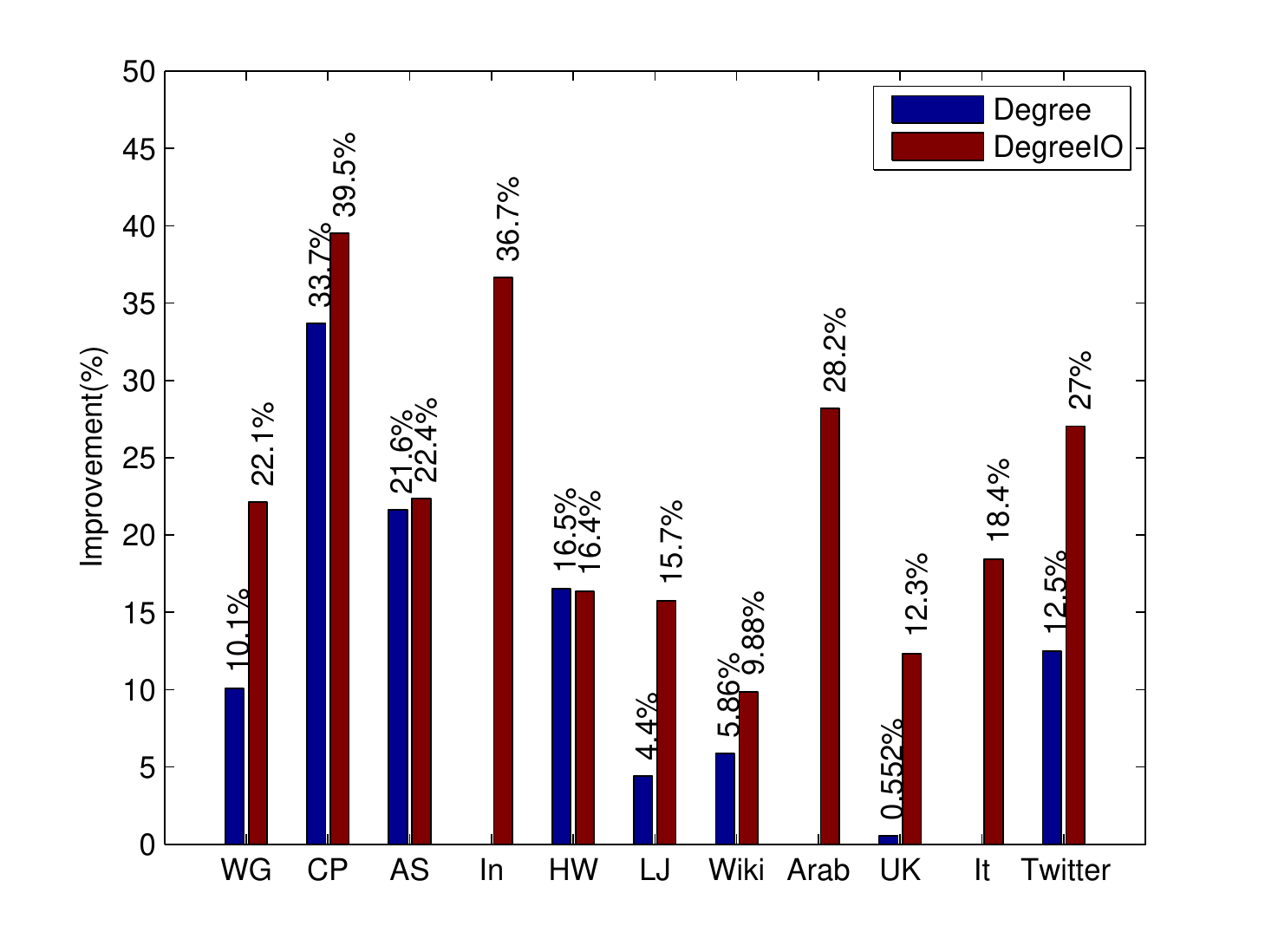}}
\subfigure[DFS order]{\includegraphics[width=0.32\textwidth]{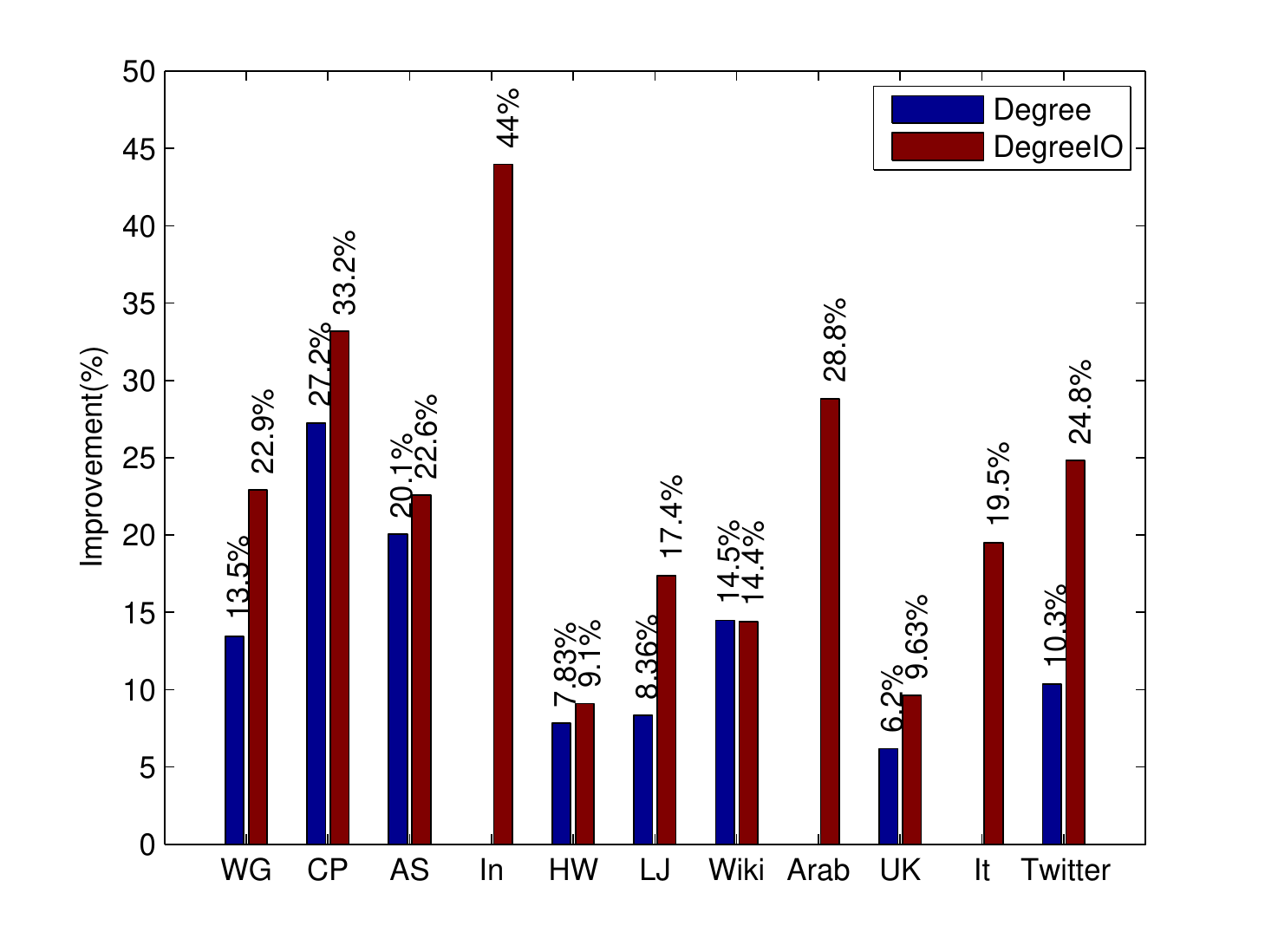}}
\vskip -0.2cm
\caption{\small Replication factor~(first row) and improvement~(second row) on real-world graphs in different stream orders.}
\label{fig:real}
\end{figure*}

In Figure~\ref{fig:sync}~{(a), (b), (c), (d), (e) and (f)}, we compared the replication factor on synthetic graphs with different power-law constants in different orders by different algorithms. In all the experiments the number of partitions is 48.
Figure~\ref{fig:sync_improvement}~{(a), (b), (c), (d), (e) and (f)} shows the improvement of the algorithms \textit{Degree} and \textit{DegreeIO} compared to \textit{Balance(PowerGraph)}.

\begin{figure*}[htb]
\centering
\subfigure[Rnd order]{\includegraphics[width=0.32\textwidth]{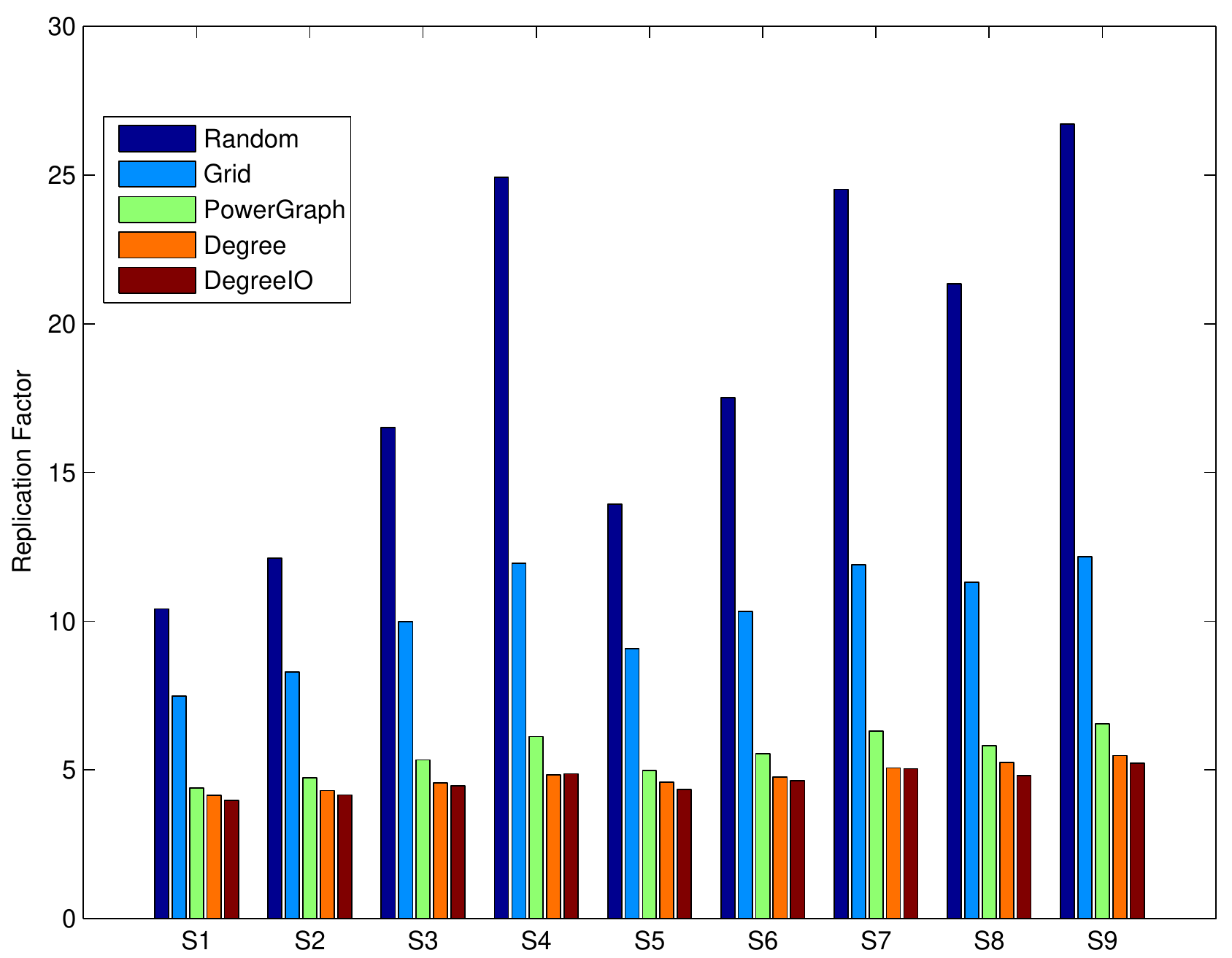}}
\subfigure[BFS order]{\includegraphics[width=0.32\textwidth]{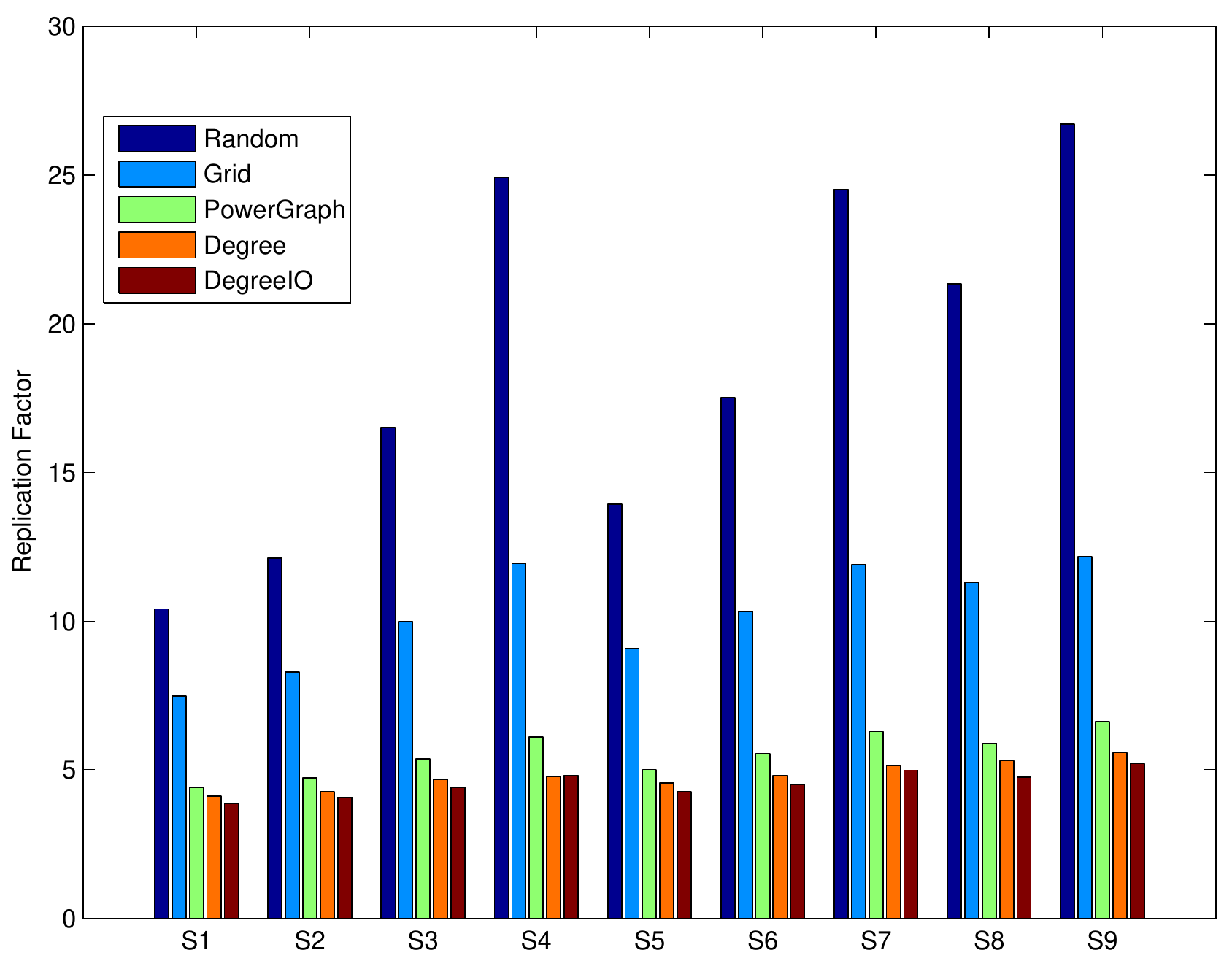}}
\subfigure[DFS order]{\includegraphics[width=0.32\textwidth]{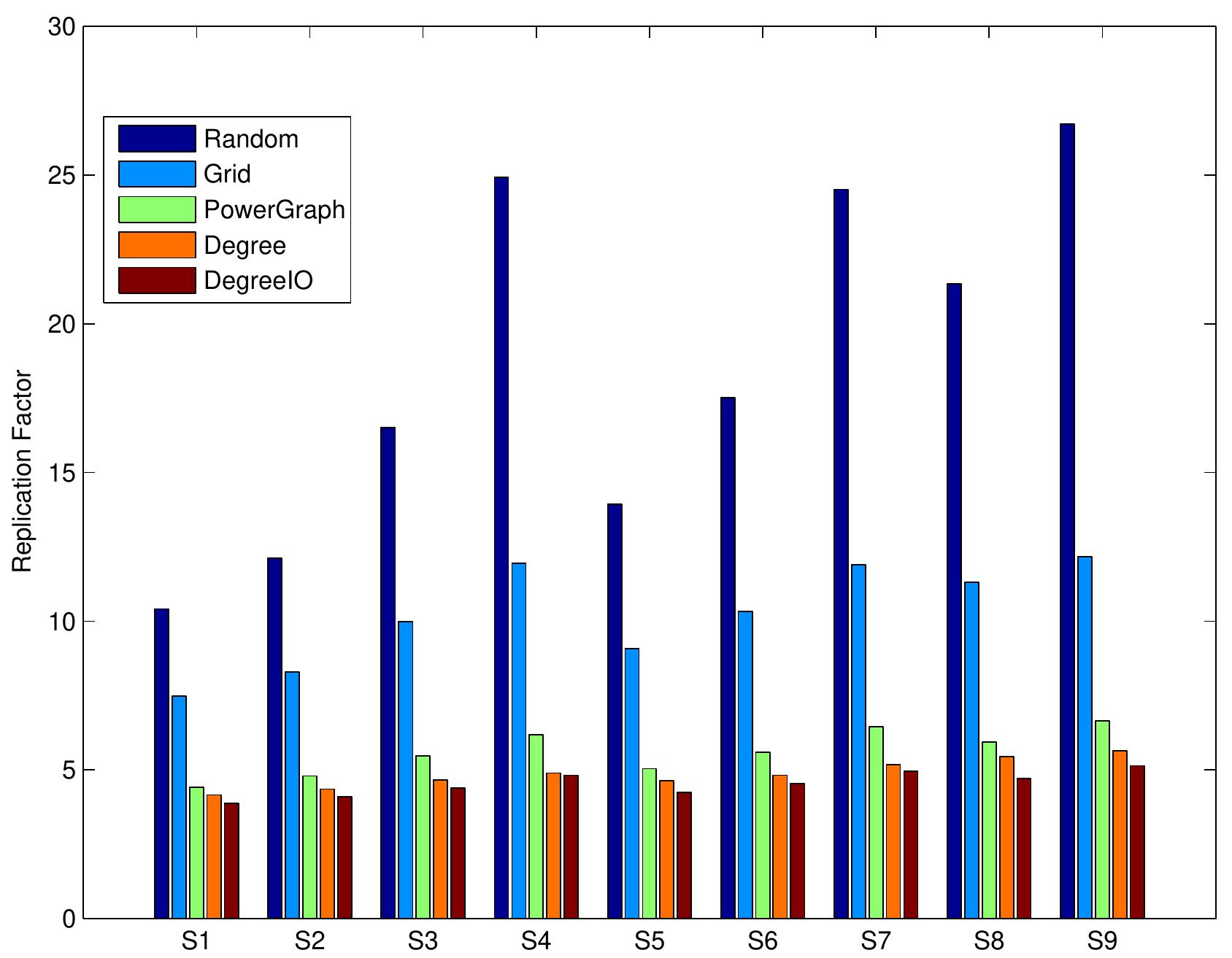}}
\\
\subfigure[Rnd order]{\includegraphics[width=0.32\textwidth]{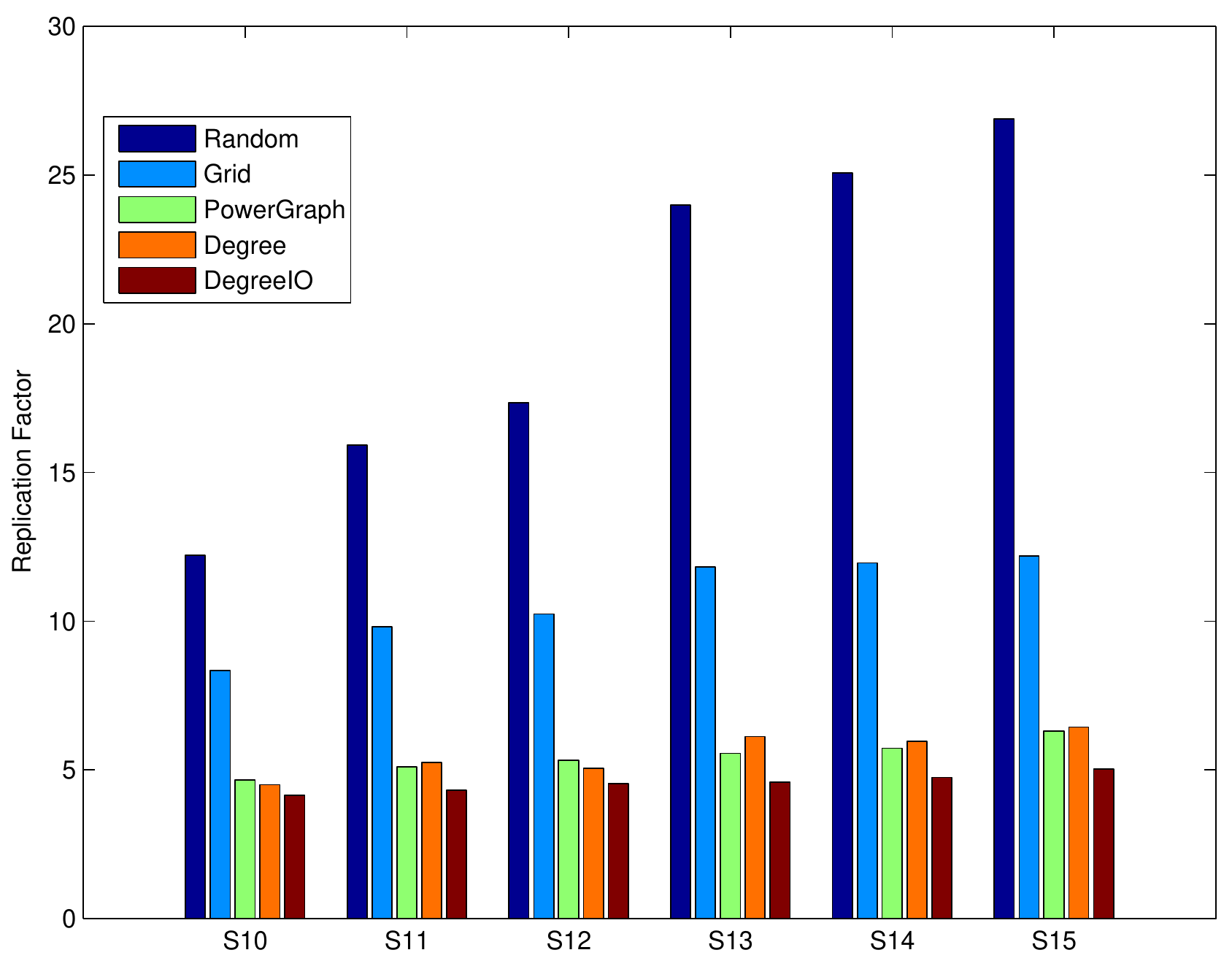}}
\subfigure[BFS order]{\includegraphics[width=0.32\textwidth]{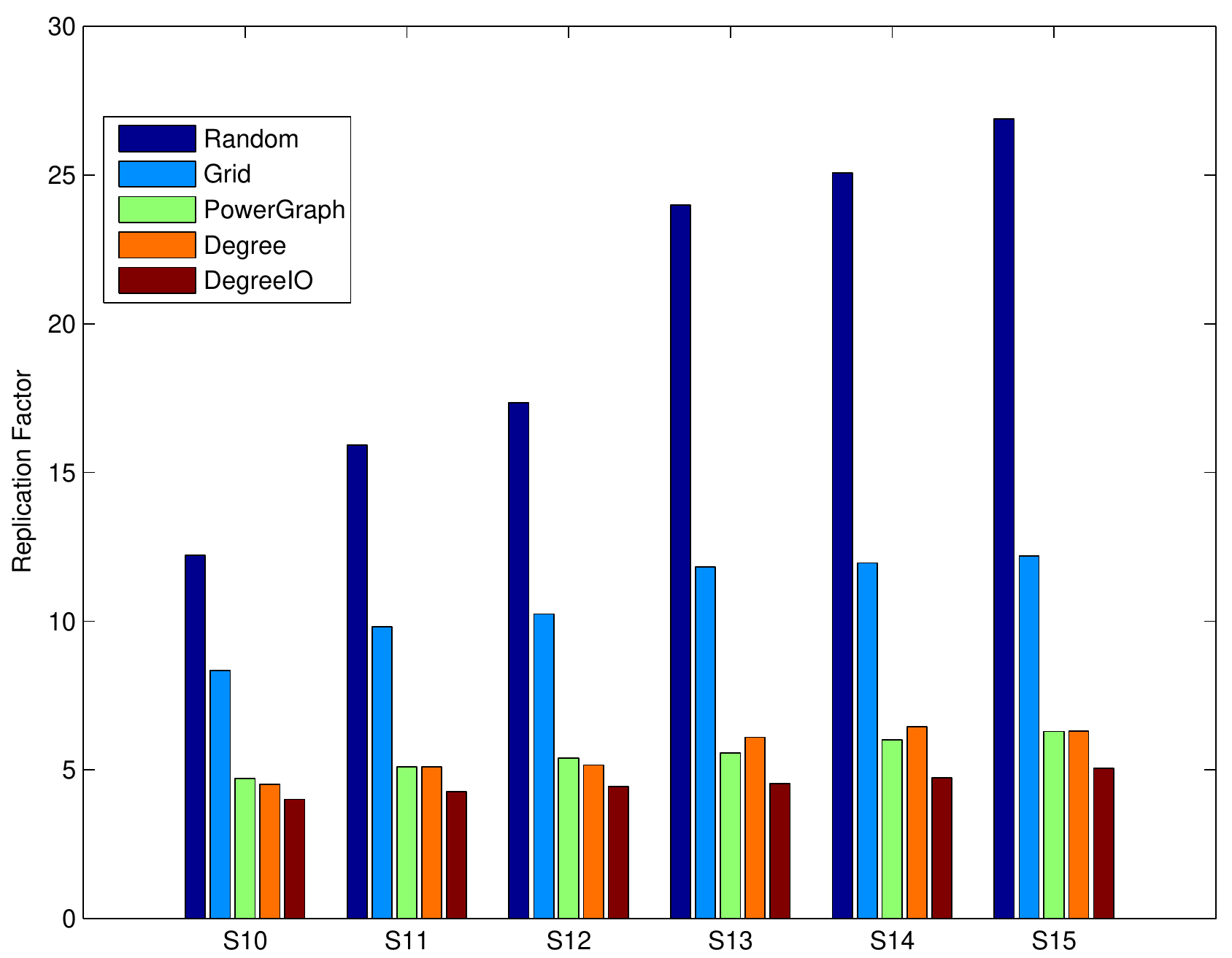}}
\subfigure[DFS order]{\includegraphics[width=0.32\textwidth]{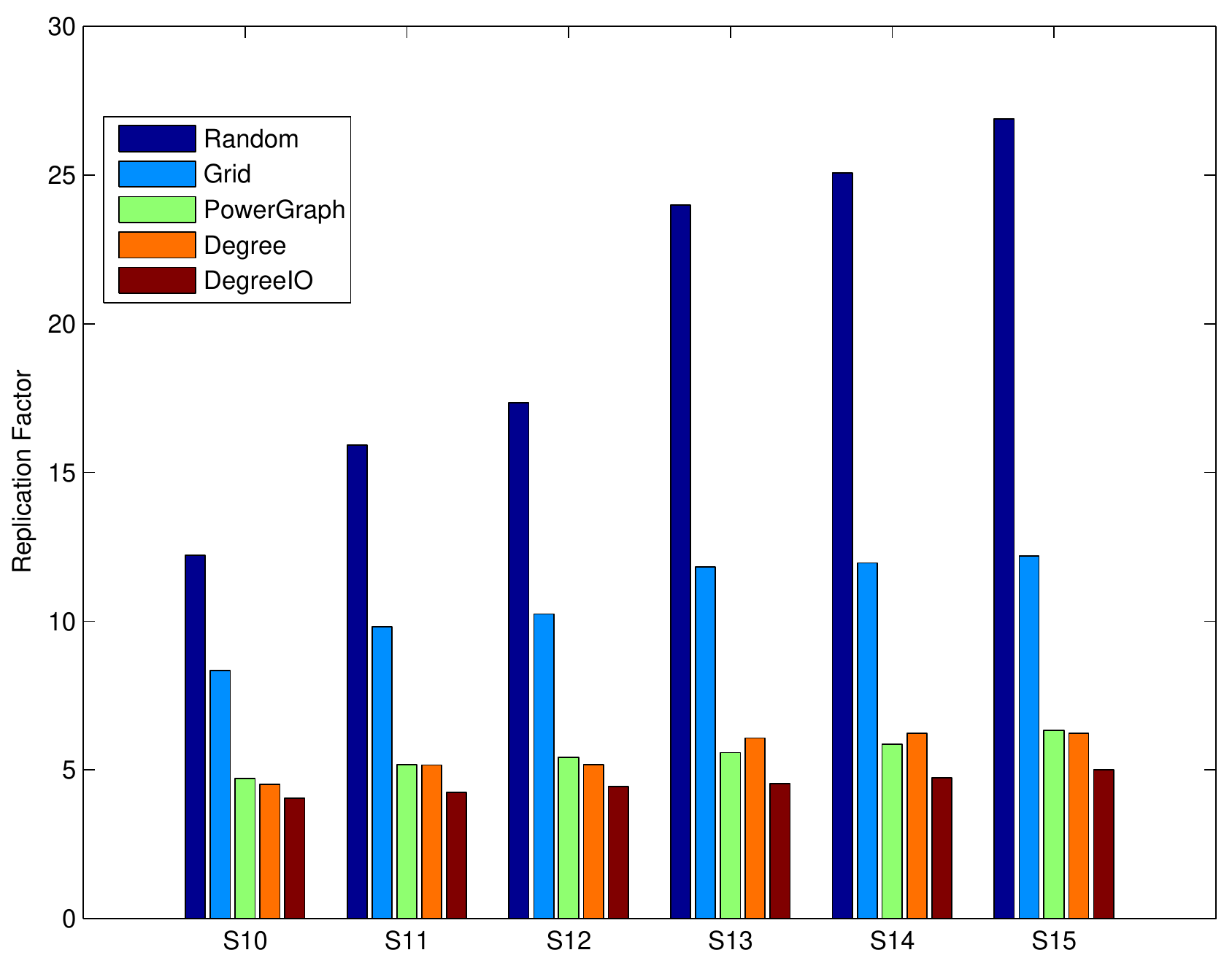}}
\vskip -0.2cm
\caption{\small Replication factor on synthetic graphs in different stream orders.}
\label{fig:sync}
\end{figure*}

\begin{figure*}[htb]
\centering
\subfigure[Rnd order]{\includegraphics[width=0.32\textwidth]{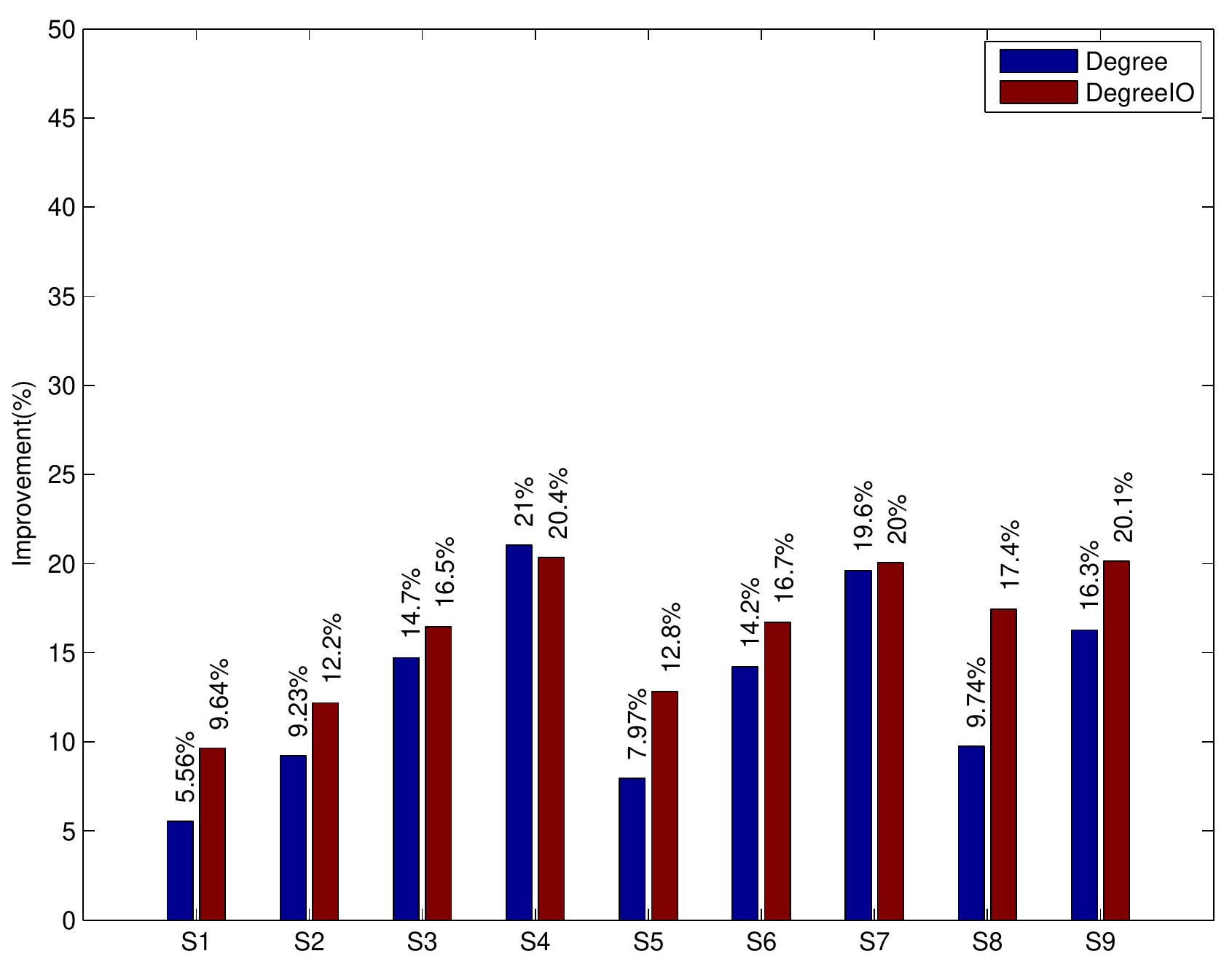}}
\subfigure[BFS order]{\includegraphics[width=0.32\textwidth]{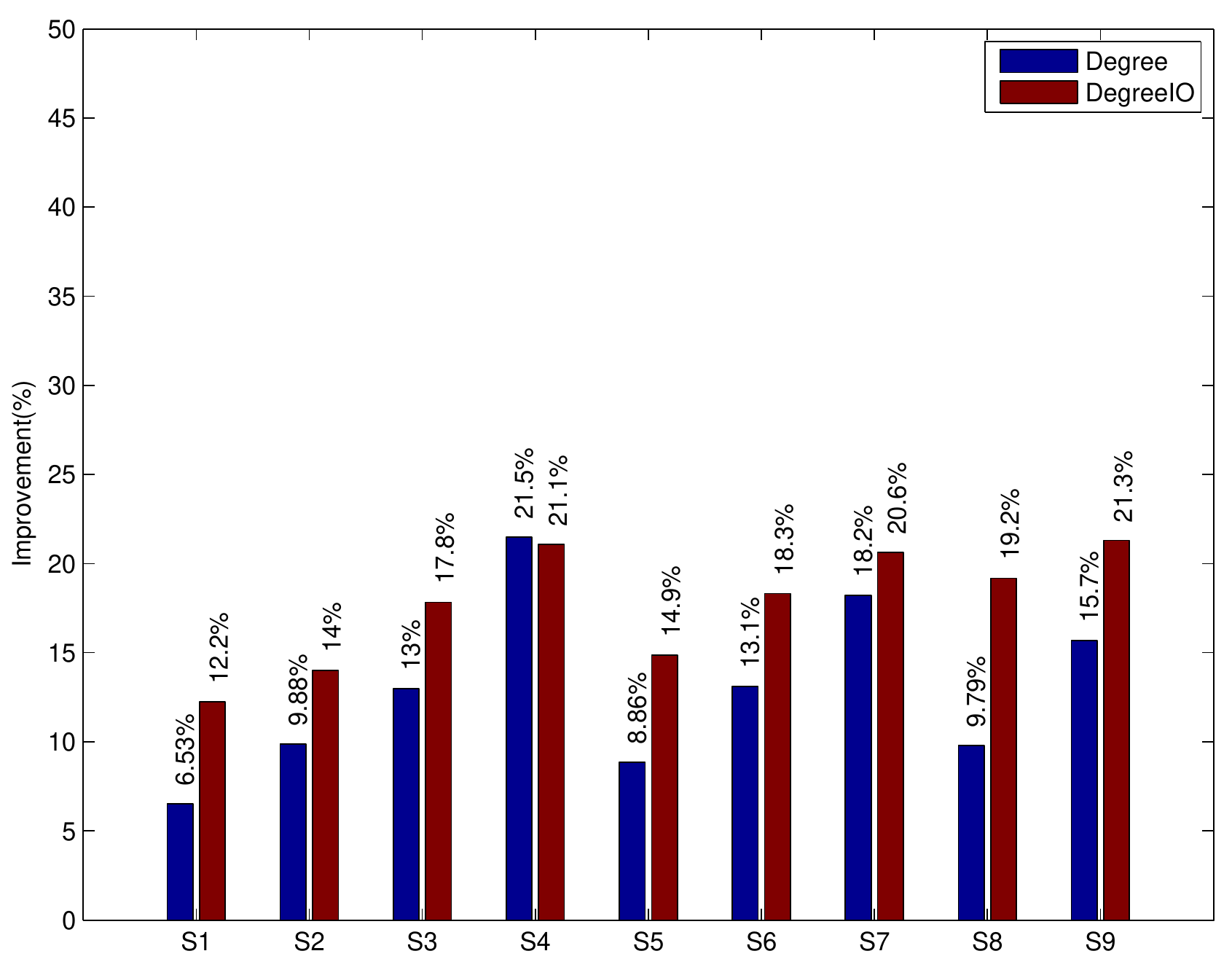}}
\subfigure[DFS order]{\includegraphics[width=0.32\textwidth]{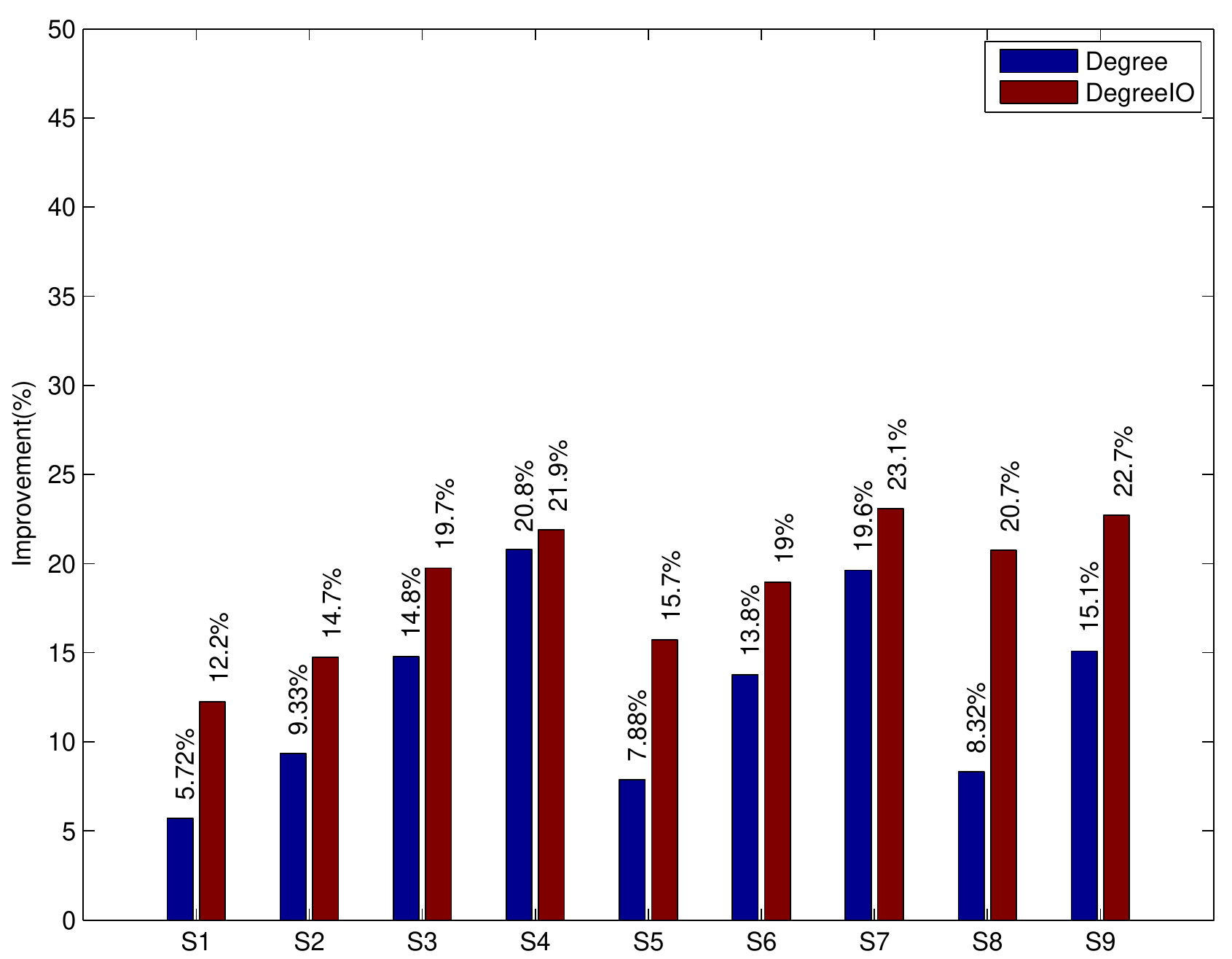}}
\\
\subfigure[Rnd order]{\includegraphics[width=0.32\textwidth]{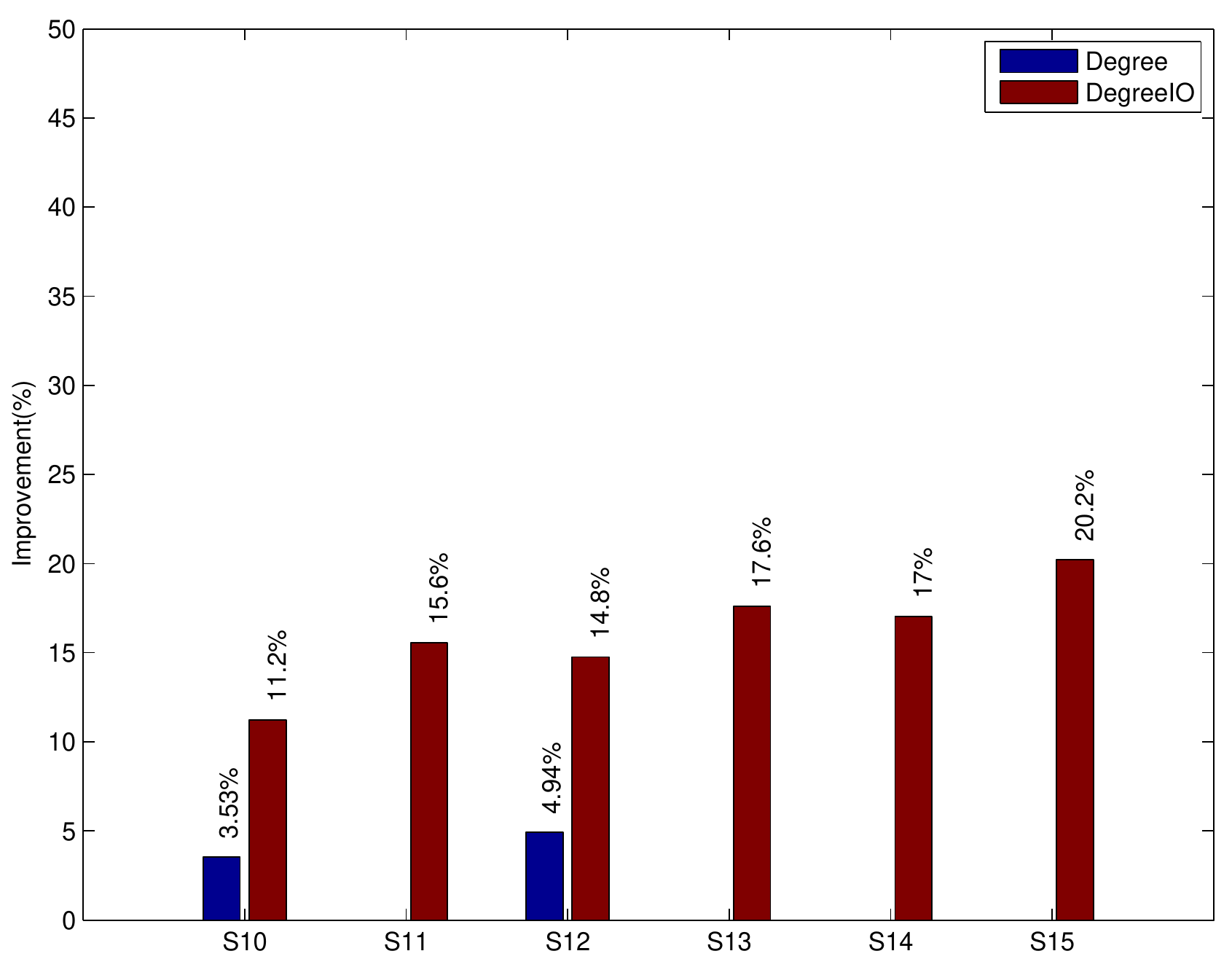}}
\subfigure[BFS order]{\includegraphics[width=0.32\textwidth]{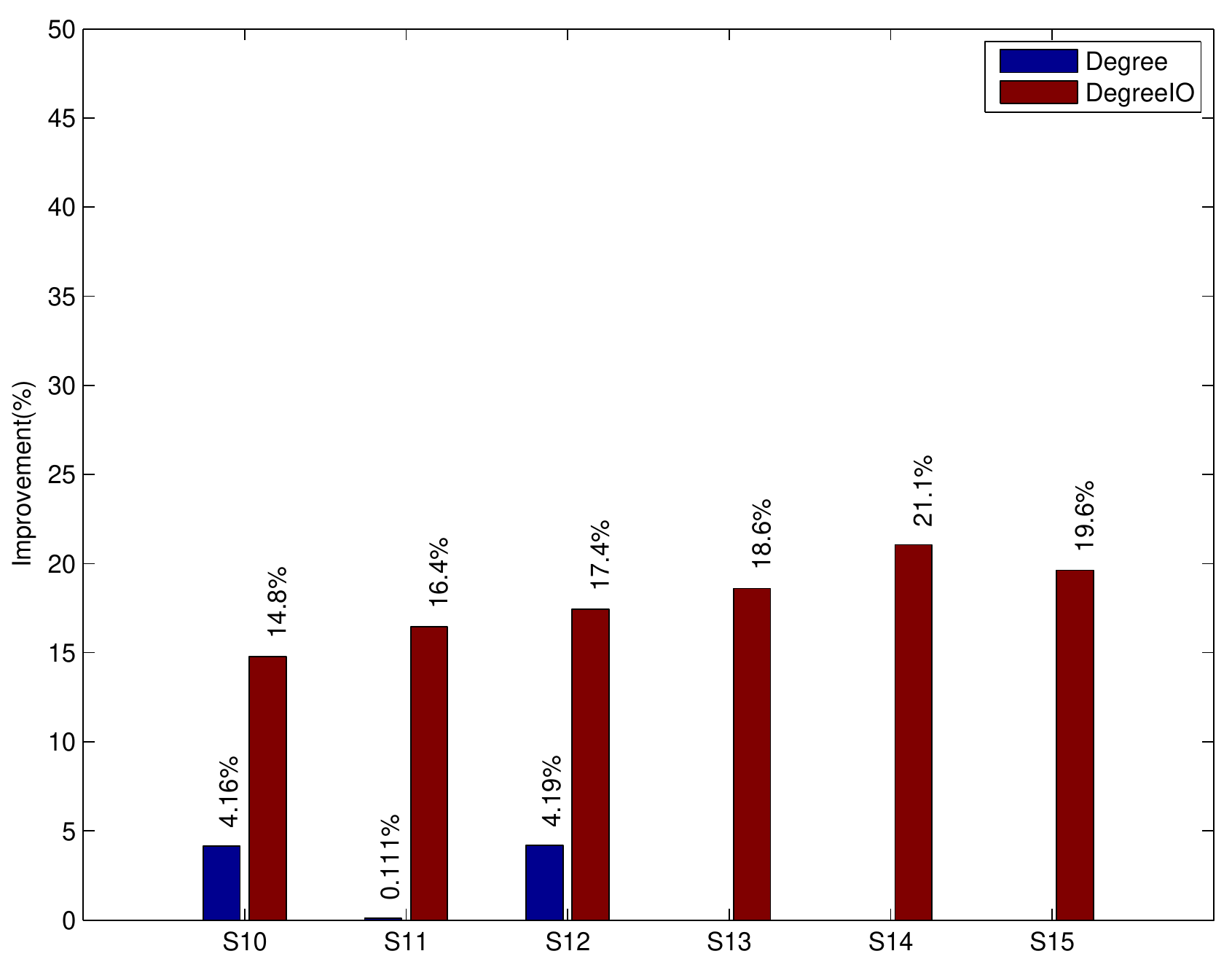}}
\subfigure[DFS order]{\includegraphics[width=0.32\textwidth]{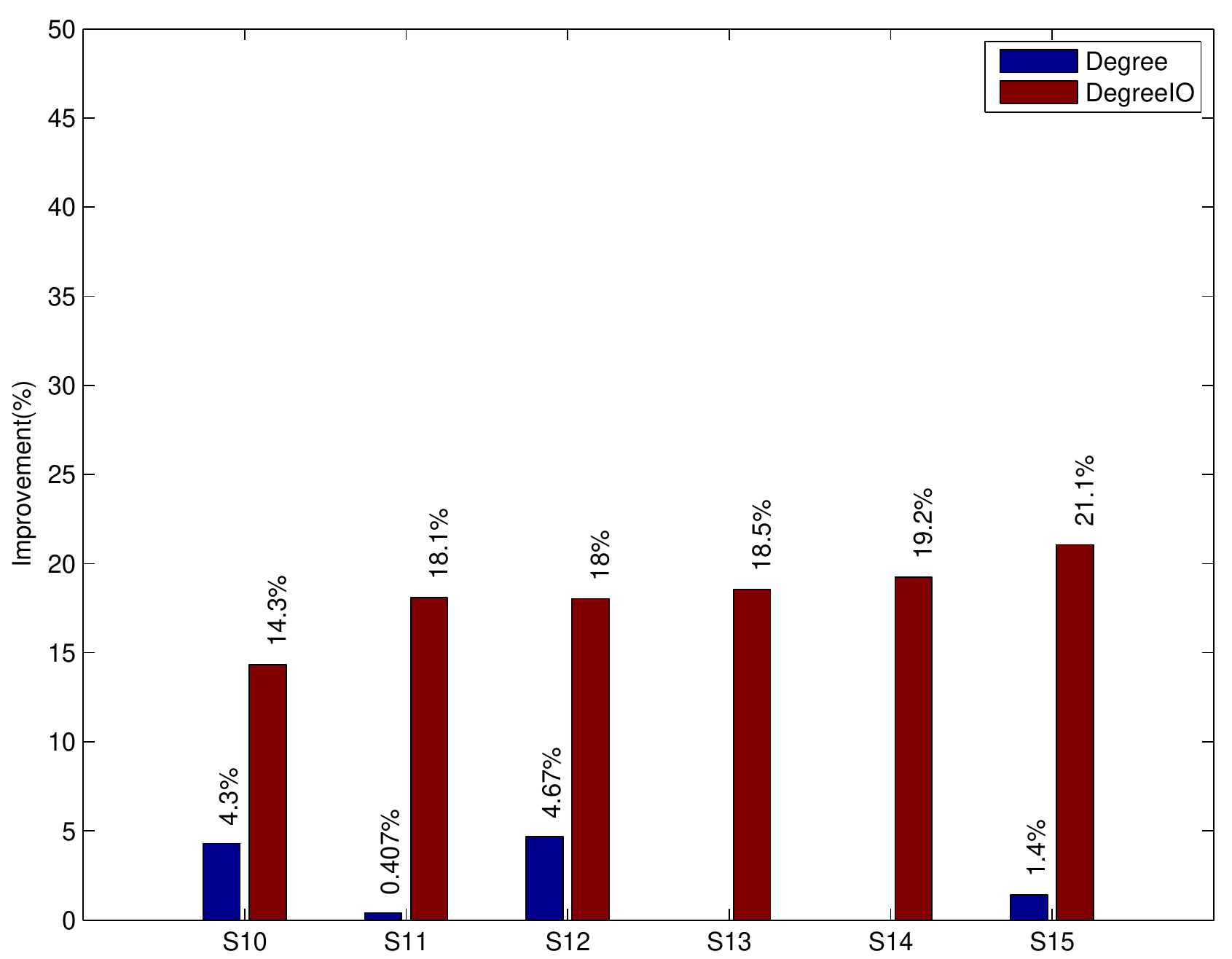}}
\vskip -0.2cm
\caption{\small Improvement on synthetic graphs in different stream orders~(compared with PowerGraph).}
\label{fig:sync_improvement}
\end{figure*}

In Table~\ref{imbalancetable1}, Table~\ref{imbalancetable2} and Table~\ref{imbalancetable3}, we compared the imbalance factor on different real-world graphs in different orders by different algorithms. In all the experiments the number of partitions is 48. The imbalance factors of our algorithms are a little bit higher but still tiny. Most of them are smaller than 1.01.
\begin{table}[htbp]
\caption{Imbalance factor on real-world graphs in random order (Here PowerGraph is abbreviated to PowerG. for a better layout.)}
\centering
\small
\begin{tabular}{|l|l|l|l|l|l|} \hline
			& \multicolumn{5}{|c|}{Method} \\ \hline
Graph		& Random	& Grid		& PowerG.	& Degree	& DegreeIO \\ \hline
WG    & 1.00725 & 1.04038 & 1.00024 & 1.00044 & 1.00041 \\
CP   & 1.00341 & 1.01943 & 1.00015 & 1.0005  & 1.00122 \\
As & 1.00408 & 1.06652 & 1.00005 & 1.00024 & 1.00018 \\
In    & 1.00394 & 1.02741 & 1.00075 & 1.01349 & 1.00298 \\
HW  & 1.00227 & 1.02349 & 1.00002 & 1.00008 & 1.00008 \\
LJ   & 1.00173 & 1.00703 & 1.00005 & 1.00006 & 1.00005 \\
Wiki       & 1.00145 & 1.03101 & 1.00002 & 1.0001  & 1.00003 \\
Arab     & 1.00057 & 1.01618 & 1.00007 & 1.00055 & 1.00006 \\
UK         & 1.0005  & 1.02545 & 1.00002 & 1.00144 & 1.00003 \\
It    & 1.00043 & 1.01874 & 1.00003 & 1.00031 & 1.00005 \\
Tw    & 1.00032 & 1.04214 & 1       & 1       & 1     	\\
\hline
\end{tabular}
\label{imbalancetable1}
\end{table}

\begin{table}[htbp]
\caption{Imbalance factor on real-world graphs in BFS order (Here PowerGraph is abbreviated to PowerG. for a better layout.)}
\centering
\small
\begin{tabular}{|l|l|l|l|l|l|} \hline
			& \multicolumn{5}{|c|}{Method} \\ \hline
Graph		& Random	& Grid		& PowerG.	& Degree	& DegreeIO \\ \hline
WG   & 1.00725 & 1.04038 & 1.00015 & 1.00036 & 1.10936 \\
CP  & 1.00341 & 1.01943 & 1.00007 & 1.00032 & 1.01744 \\
As & 1.00408 & 1.06652 & 1.00003 & 1.00021 & 1.00022 \\
In    & 1.00394 & 1.02741 & 1.01396 & 1.0163  & 1.04447 \\
HW  & 1.00227 & 1.02349 & 1.0245  & 1.0399  & 1.00011 \\
LJ   & 1.00173 & 1.00703 & 1.00001 & 1.00721 & 1.00006 \\
Wiki       & 1.00145 & 1.03101 & 1       & 1.00134 & 1.00001 \\
Arab    & 1.00057 & 1.01618 & 1.0001  & 1.00462 & 1.00848 \\
UK         & 1.0005  & 1.02545 & 1.00046 & 1.00069 & 1.00002 \\
It    & 1.00043 & 1.01874 & 1.00024 & 1.01175 & 1.06839 \\
Tw    & 1.00032 & 1.04214 & 1       & 1       & 1  \\
\hline
\end{tabular}
\label{imbalancetable2}
\end{table}

\begin{table}[htbp]
\caption{Imbalance factor on real-world graphs in DFS order (Here PowerGraph is abbreviated to PowerG. for a better layout.)}
\centering
\small
\begin{tabular}{|l|l|l|l|l|l|} \hline
			& \multicolumn{5}{|c|}{Method} \\ \hline
Graph		& Random	& Grid		& PowerG.	& Degree	& DegreeIO \\ \hline
WG    & 1.00725 & 1.04038 & 1.00014 & 1.00043 & 1.00024 \\
CP  & 1.00341 & 1.01943 & 1.00009 & 1.00027 & 1.00035 \\
As & 1.00408 & 1.06652 & 1.00002 & 1.00015 & 1.00019 \\
In    & 1.00394 & 1.02741 & 1.0161  & 1.01655 & 1.04168 \\
HW  & 1.00227 & 1.02349 & 1.03225 & 1.04545 & 1.00177 \\
LJ   & 1.00173 & 1.00703 & 1.00003 & 1.00534 & 1.00007 \\
Wiki       & 1.00145 & 1.03101 & 1       & 1.00002 & 1.00001 \\
Arab     & 1.00057 & 1.01618 & 1.0003  & 1.00787 & 1.02141 \\
UK         & 1.0005  & 1.02545 & 1.00022 & 1.00568 & 1.00002 \\
It    & 1.00043 & 1.01874 & 1.00018 & 1.047   & 1.00005 \\
Tw    & 1.00032 & 1.04214 & 1       & 1       & 1     	\\
\hline
\end{tabular}
\label{imbalancetable3}
\end{table}

In Figure~\ref{fig:twitter}~{(a), (b) and (c)}, we evaluate our algorithm on Twitter-2010 in different order to compare the sensibility to the number of partitions. Our algorithms is always better and the improvement gets more obvious when the number of partitions increases.
\begin{figure*}[!htb]
\subfigure[Rnd order]{\includegraphics[width=0.32\textwidth]{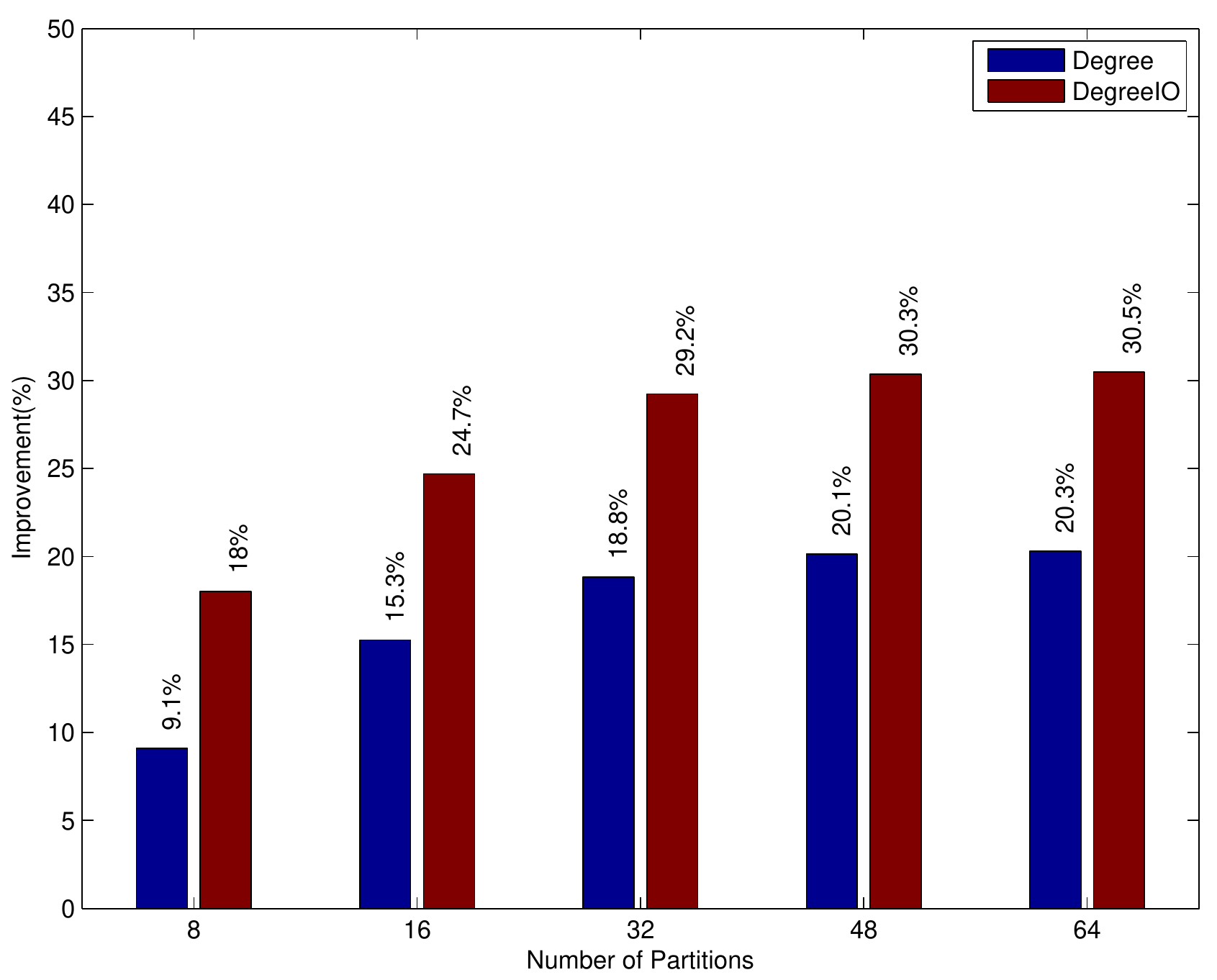}}
\subfigure[BFS order]{\includegraphics[width=0.32\textwidth]{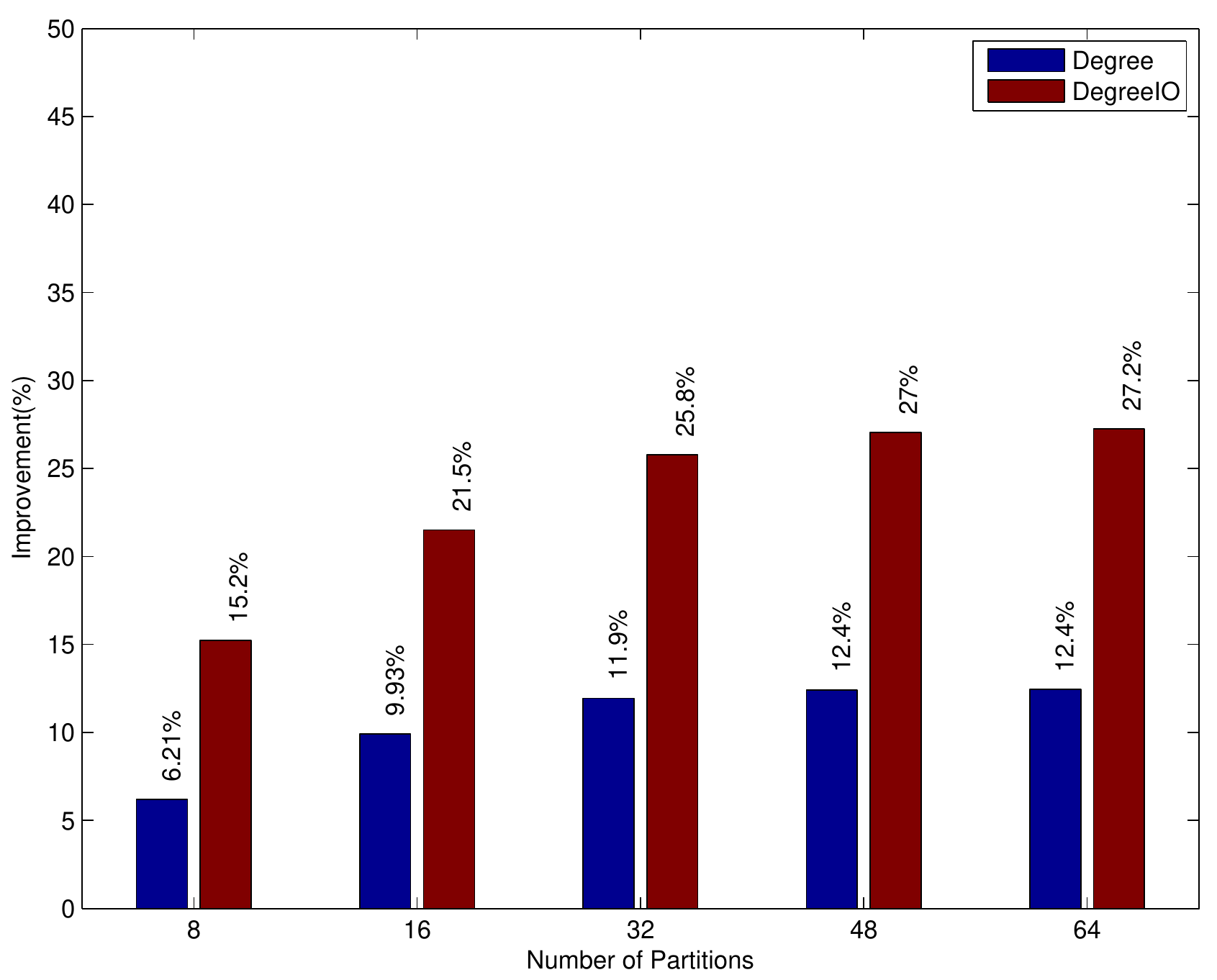}}
\subfigure[DFS order]{\includegraphics[width=0.32\textwidth]{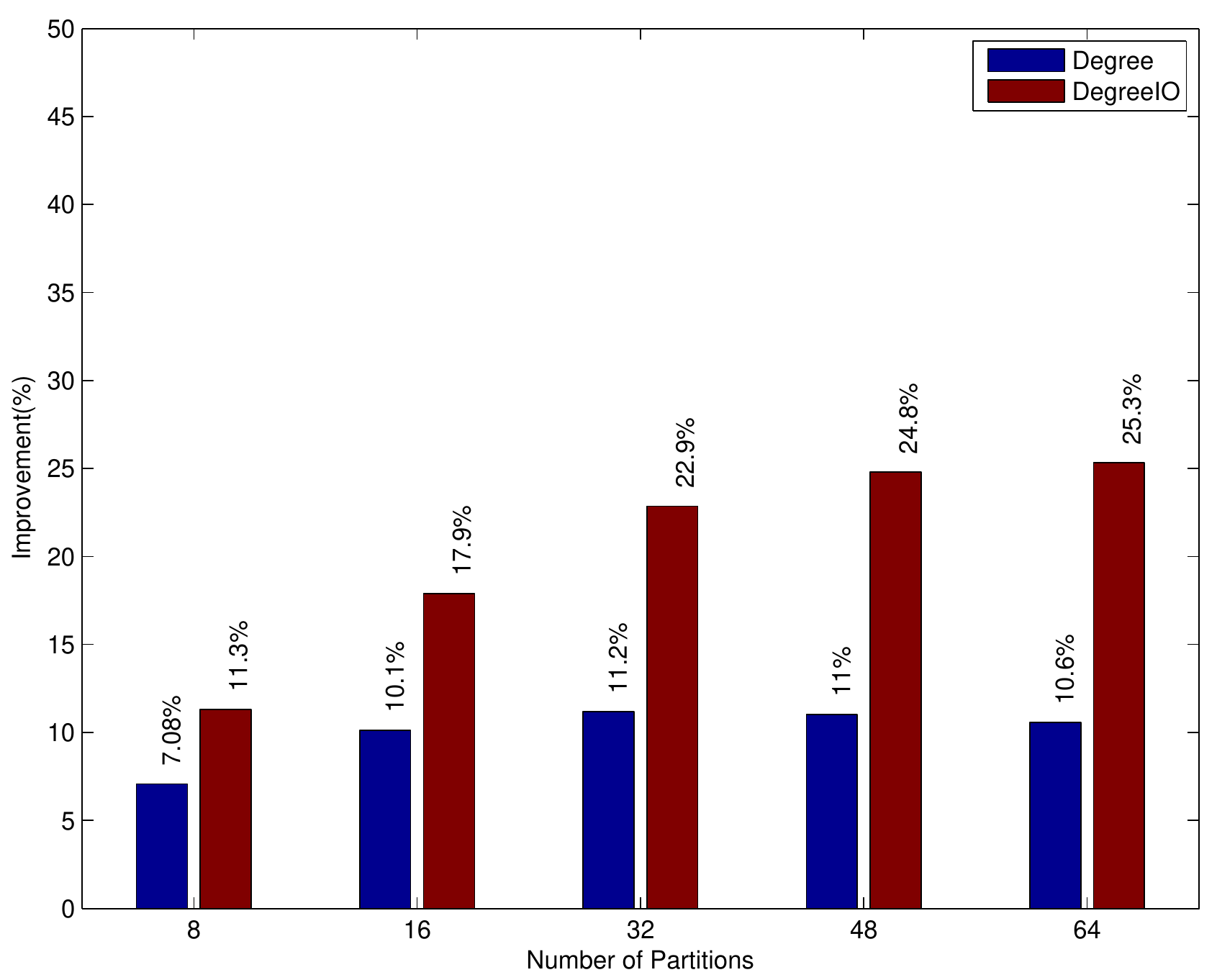}}
\vskip -0.2cm
\caption{\small Improvement on Twitter-2010 in different stream orders~(compared with PowerGraph).}
\label{fig:twitter}
\end{figure*}


\subsection{Discussion}

Taking into account all the experiments, our algorithm \textit{DegreeIO} is better than \textit{Balance(PowerGraph)} in all datasets. And \textit{Degree} improves the performance of \textit{Balance(PowerGraph)} in most cases.
For the real-world graphs, we can see that the improvement on WebGoogle is not very obvious. The reason may be that the WebGoogle graph is not so skewed and the power-law constant may be high. For the largest dataset Twitter-2010, \textit{Degree} gets an improvement for about 15\% to 20\% while \textit{DegreeIO} gets nearly 25\% to 30\%. Our algorithms produce more improvement than the algorithm introduced by PowerGraph\cite{gonzalez2012powergraph} when the graphs are more skewed in the degree distribution which can be shown in the experiment on synthetic graphs. This kind of phenomenon is a consequence of the motivation of our algorithm. Recall that our algorithm makes use of the power-law degree distribution to further reduce the vertex-cut. So our algorithm will not produce an obvious improvement if the distribution of degrees are nearly uniform, and it works better when the graph is more skewed.
Finally the experiments on Twitter-2010 shows good scalability for our algorithms.

Note that the algorithm \textit{Degree} does not always produce a better replication factor than \textit{Balance(PowerGraph)}. The reason is that \textit{Degree} estimates the degree distribution only by the in-degree. Thus if the out-degree is more skewed than the in-degree, for example In-2004, Arabic-2005, It-2004 and the second part of the synthetic graph collection, the estimation will be wrong which will result in bad partitioning. And \textit{DegreeIO} estimate the degree distribution by both in-degree and out-degree, so the performance is better in such datasets. For example, the improvement of Arabic-2005 is nearly 35\%.


\section{Conclusion}

In this paper we have studied streaming graph partitioning by vertex-cut. In particular,  we have proposed a novel streaming graph partitioning method for natural graphs by adopting vertex-cut strategy and called it \emph{S-PowerGraph}. To the best of our knowledge, this is the first work to systematically study the effectiveness of vertex-cut in streaming graph partitioning. Experiments on large graphs show that our method is more suitable for partitioning skewed natural graphs than the previous approaches with an acceptable imbalance factor. Our method can be used to partition the raw graph data from a web crawler or the graph data loaded from a single storage device. In our future work, we would like to pursued more variants of the proposed strategy theoretically and empirically.

\bibliographystyle{abbrv}
\bibliography{spowergraph}

\end{document}